\documentclass[a4paper,UKenglish,autoref,numberwithinsect]{mylipics-v2019}
\usepackage[T1]{fontenc}
\usepackage[utf8]{inputenc}

\newcounter{blubber}
\newenvironment{longitemslist}
{\begin{list}
  {\textsf{\textbf{\arabic{blubber}.}}}
  {\usecounter{blubber}
   \setlength{\leftmargin}{0pt}
    \setlength{\parsep}{0pt}
    \setlength{\itemindent}{4ex}
    \setlength{\itemsep}{2pt}
  }
}
{\end{list}}

\usepackage{lmodern}
\usepackage{microtype}
\usepackage{csquotes}

\usepackage{amssymb,amsmath,amsbsy} 
\usepackage{mathtools}
\usepackage{stmaryrd}
\SetSymbolFont{stmry}{bold}{U}{stmry}{m}{n}
\usepackage{dsfont}

\usepackage{xspace}

\usepackage{hyperref}
\newenvironment{defn}[1][]{\begin{definition}[#1]\upshape}{\end{definition}}
\usepackage{ifthen}
\newboolean{full}
\setboolean{full}{true}

\newcommand{\takeout}[1]{\empty}

\usepackage{tikz}
\usetikzlibrary{cd} 

\usepackage[all]{xy}
\SelectTips{cm}{}

\usepackage{dsfont}

\usepackage{cite}

\usepackage[final]{fixme}

\addto\extrasUKenglish{
}

\newcommand{\myparagraph}[1]{\medskip\par\noindent\textbf{\textsf{#1}}\hspace{6pt}}

\usepackage{seqsplit}
\usepackage{xstring}

\makeatletter
\usepackage[notcite,notref,final]{showkeys}%

\fxuselayouts{inline,nomargin}
\fxusetheme{color}
\FXRegisterAuthor{ls}{als}{LS}
\FXRegisterAuthor{sm}{asm}{SM}
\FXRegisterAuthor{ud}{aud}{UD}

\tikzcdset{
  arrow style=tikz,
  diagrams={>=Straight Barb,column sep=normal,row sep=normal}
}

\newcommand{\nth}{\ensuremath{n^\text{th}}}
\DeclareMathOperator{\Id}{Id}
\DeclareMathOperator{\id}{id}
\DeclareMathOperator{\inl}{inl}

\newcommand{\T}{\ensuremath{\mathcal T}\xspace}

\newcommand{\cat}[1]{\ensuremath{\mathbf{#1}}\xspace}
\newcommand\Set{\cat{Set}}
\newcommand{\dist}{\ensuremath{\mathcal D}\xspace}

\newcommand{\unit}{\ensuremath{\eta}\xspace}
\newcommand{\mult}{\ensuremath{\mu}\xspace}

\renewcommand\epsilon\varepsilon

\newcommand{\truthstruct}{o}
\newcommand{\GTh}{\mathbb{T}}
\newcommand{\Act}{\mathcal{A}}
\newcommand{\Down}{\Pfin^\downarrow}
\newcommand{\trdiamond}[1]{\Diamond_{#1}}
\newcommand{\bbrack}[1]{{\llbracket #1 \rrbracket}}
\newcommand{\Sem}[1]{\bbrack{#1}}
\newcommand{\into}{\hookrightarrow}
\newcommand{\contrapow}{\mathcal{Q}}
\newcommand{\Op}{\mathit{op}}
\newcommand{\FA}{\mathfrak{A}}
\newcommand{\BBM}{M}
\newcommand{\by}[1]{(\text{#1})}
\newcommand{\PropOps}{\mathcal{O}}
\newcommand{\BC}{\mathbf{C}}
\newcommand{\Pow}{\mathcal{P}}
\newcommand{\Pfin}{\Pow_\omega}
\newcommand{\Dist}{\mathcal{D}}

\newcommand{\CT}{\ensuremath{\mathcal {CT}}\xspace}
\newcommand{\RT}{\ensuremath{\mathcal {RT}}\xspace}
\newcommand{\FT}{\ensuremath{\mathcal{FT}}\xspace}%
\newcommand{\R}{\ensuremath{\mathcal R}}
\newcommand{\F}{\ensuremath{\mathcal F}}

\title{Graded Monads and Graded Logics\protect\\for the Linear Time --
  Branching Time Spectrum}

\titlerunning{Graded Monads and Graded Logics for the Linear Time -- Branching Time Spectrum}

\author{Ulrich Dorsch, Stefan Milius and Lutz Schröder}%
  {Friedrich-Alexander-Universität Erlangen-Nürnberg, Germany}{}{}%
  {}

\authorrunning{U. Dorsch and S. Milius and L. Schröder}

\Copyright{Ulrich Dorsch and Stefan Milius and Lutz Schröder}
\relatedversion{Full version with all proof details
    available at \url{https://arxiv.org/abs/1812.01317}}
\funding{This work forms part of the DFG-funded project COAX (MI~717/5-2 and SCHR~1118/12-2)}


\nolinenumbers 

\begin{document}

\maketitle

\begin{abstract}
  State-based models of concurrent systems are traditionally
  considered under a variety of notions of process equivalence. In the
  case of labelled transition systems, these equivalences range from
  trace equivalence to (strong) bisimilarity, and are organized in
  what is known as the linear time -- branching time spectrum. A
  combination of universal coalgebra and graded monads provides a
  generic framework in which the semantics of concurrency can be
  parametrized both over the branching type of the underlying
  transition systems and over the granularity of process
  equivalence. We show in the present paper that this framework of
  \emph{graded semantics} does subsume the most important equivalences
  from the linear time -- branching time spectrum. An important
  feature of graded semantics is that it allows for the principled
  extraction of characteristic modal logics. We have established
  invariance of these \emph{graded logics} under the given graded
  semantics in earlier work; in the present paper, we extend the
  logical framework with an explicit propositional layer and provide a
  generic expressiveness criterion that generalizes the classical
  Hennessy-Milner theorem to coarser notions of process
  equivalence. We extract graded logics for a range of graded
  semantics on labelled transition systems and probabilistic systems,
  and give exemplary proofs of their expressiveness based on our
  generic criterion.
\end{abstract}

\section{Introduction}

State-based models of concurrent systems are standardly considered
under a wide range of system equivalences, typically located between
two extremes respectively representing \emph{linear time} and
\emph{branching time} views of system evolution. Over labelled
transition systems, one specifically has the well-known \emph{linear
  time -- branching time spectrum} of system equivalences between
trace equivalence and
bisimilarity~\cite{vanglabbeek2001linear}. Similarly, e.g.\
probabilistic automata have been equipped with various semantics
including strong bisimilarity~\cite{LarsenSkou91}, probabilistic
(convex) bisimilarity~\cite{SegalaLynch94}, and distribution
bisimilarity (e.g.~\cite{DengEA08,DoyenEA08}). New equivalences keep
appearing in the literature, e.g.~for non-deterministic probabilistic
systems~\cite{BonchiEA19,vanHeerdtEA18}.

This motivates the search for unifying principles that allow for a
generic treatment of process equivalences of varying degrees of
granularity and for systems of different branching types
(non-deterministic, probabilistic etc.). As regards the variation of
the branching type, universal coalgebra~\cite{Rutten00} has emerged as
a widely-used uniform framework for state-based systems covering a
broad range of branching types including besides non-deterministic and
probabilistic, or more generally weighted, branching also, e.g.,
alternating, neighbourhood-based, or game-based systems. It is based
on modelling the system type as an endofunctor on some base category,
often the category of sets.

Unified treatments of system equivalences, on the other hand, are so
far less well-established, and their applicability is often more
restricted. Existing approaches include coalgebraic trace semantics in
Kleisli~\cite{HasuoEA07} and Eilenberg-Moore
categories~\cite{KissigKurz10,JacobsEA12,sbbr13,bms13,BonchiEA19,vanHeerdtEA18},
respectively.  Both semantics are based on decomposing the coalgebraic
type functor into a monad, the \emph{branching type}, and a functor,
the \emph{transition type} (in different orders), and require suitable
distributive laws between these parts; correspondingly, they grow
naturally out of the functor but on the other hand apply only to
functors that admit the respective form of decomposition. In the
present work, we build on a more general approach introduced by
Pattinson and two of us, based on mapping the coalgebraic type functor
into a \emph{graded monad}~\cite{MiliusEA15}. Graded monads correspond
to algebraic theories where operations come with an explicit notion of
\emph{depth}, and allow for a stepwise evaluation of process
semantics. Maybe most notably, graded monads systematically support a
reasonable notion of \emph{graded logic} where modalities are
interpreted as \emph{graded algebras} for the given graded monad. This
approach applies to all cases covered in the mentioned previous
frameworks, and additional cases that do not fit any of the earlier
setups. We emphasize that graded monads are geared towards
\emph{inductively} defined equivalences such as finite trace semantics
and finite-depth bisimilarity; we leave a similarly general treatment
of infinite-depth equivalences such as infinite trace equivalence and
unbounded-depth bisimilarity to future work. To avoid excessive
verbosity, we restrict to models with finite branching
throughout. Under finite branching, finite-depth equivalences
typically coincide with their infinite-depth counterparts, e.g.\
states of finitely branching labelled transition systems are bisimilar
iff they are finite-depth bisimilar, and infinite-trace equivalent iff
they are finite-trace equivalent.

Our goal in the present work is to illustrate the level of generality
achievable by means of graded monads in the dimension of system
equivalences. We thus pick a fixed coalgebraic type, that of labelled
transition systems, and elaborate how a number of equivalences from
the linear time -- branching time spectrum are cast as graded
monads. In the process, we demonstrate how to extract logical
characterizations of the respective equivalences from most of the
given graded monads. For the time being, none of the logics we find
are sensationally new, and in fact van Glabbeek already provides
logical characterizations in his exposition of the linear time --
branching time spectrum~\cite{vanglabbeek2001linear}; an overview of
characteristic logics for non-deterministic and probabilistic
equivalences is given by Bernardo and
Botta~\cite{bernardo-botta:characterising-logics}. The emphasis in the
examples is mainly on showing how the respective logics are developed
uniformly from general principles.

Using these examples as a backdrop, we develop the theory of graded
monads and graded logics further. In particular,
\begin{itemize}
\item we give a more economical characterization of depth-$1$ graded
  monads involving only two functors (rather than an infinite sequence
  of functors);
\item we extend the logical framework by a treatment of propositional
  operators -- previously regarded as integrated into the modalities
  -- as first class citizens; 
\item we prove, as our main technical result, a generic expressiveness
  criterion for graded logics guaranteeing that inequivalent states are
  separated by a trace formula. 
\end{itemize}
Our expressiveness criterion subsumes, at the branching-time end of
the spectrum, the classical Hennessy-Milner
theorem~\cite{HennessyMilner85} and its coalgebraic
generalization~\cite{Pattinson04,Schroder08} as well as expressiveness
of probabilistic modal logic with only
conjunction~\cite{DesharnaisEA98}; we show that it also covers
expressiveness of the respective graded logics for more coarse-grained
equivalences along the linear time -- branching time spectrum. To
illustrate generality also in the branching type, we moreover provide
an example in a probabilistic setting, specifically we apply our
expressiveness criterion to show expressiveness of a quantitative
modal logic for probabilistic trace equivalence.

\myparagraph{Related Work} Fahrenberg and
Legay~\cite{FahrenbergLegay17} characterize equivalences on the linear
time -- branching time spectrum by suitable classes of modal
transition systems. We have already mentioned previous work on
coalgebraic trace semantics in Kleisli and Eilenberg-Moore
categories~\cite{HasuoEA07,KissigKurz10,JacobsEA12,sbbr13,bms13,BonchiEA19,vanHeerdtEA18}. A
common feature of these approaches is that, more precisely speaking,
they model \emph{language} semantics rather than trace semantics --
i.e.\ they work in settings with explicit successful termination, and
consider only successfully terminating traces. When we say that graded
monads apply to all scenarios covered by these approaches, we mean
more specifically that the respective language semantics are obtained
by a further canonical quotienting of our trace
semantics~\cite{MiliusEA15}. Having said that graded monads are
strictly more general than Kleisli and Eilenberg-Moore style trace
semantics, we hasten to add that the more specific setups have their
own specific benefits including final coalgebra characterizations and,
in the Eilenberg-Moore setting, generic determinization procedures. A
further important piece of related work is Klin and Rot's method of
defining trace semantics via the choice of a particular flavour of
trace logic~\cite{KlinRot15}. In a sense, this approach is opposite to
ours: A trace logic is posited, and then two states are declared
equivalent if they satisfy the same trace formulae. In our approach
via graded monads, we instead pursue the ambition of first fixing a
semantic notion of equivalence, and then designing a logic that
characterizes this equivalence. Like Klin and Rot, we view trace
equivalence as an inductive notion, and in particular limit attention
to finite traces; coalgebraic approaches to infinite traces exist, and
mostly work within the Kleisli-style
setup~\cite{Jacobs04,Cirstea11,KerstanKonig13,Cirstea14,Cirstea15,UrabeHasuo15,Cirstea17}. Jacobs,
Levy and Rot~\cite{JacobsEA18} use corecursive algebras to provide a
unifying categorical view on the above-mentioned approaches to traces
via Kleisli- and Eilenberg-Moore categories and trace logics,
respectively. This framework does not appear to subsume the approach
via graded monads, and like for the previous approaches we are not
aware that it covers semantics from the linear time -- branching time
spectrum other than the end points (bisimilarity and trace
equivalence).

\section{Preliminaries: Coalgebra}\label{sec:prelim}

We recall basic definitions and results in \emph{(universal)
  coalgebra}~\cite{Rutten00}, a framework for the unified treatment of
a wide range of reactive systems. We write~$1=\{\star\}$ for a fixed
one-element set, and $!\colon X\to 1$ for the unique map from a set~$X$
into~$1$. We write $f\cdot g$ for the composite of maps $g\colon X\to Y$,
$f\colon Y\to Z$, and $\langle f,g\rangle\colon X\to Y\times Z$ for the pair map
$x\mapsto(f(x),g(x))$ formed from maps $f\colon X\to Y$, $g\colon X\to Z$.

Coalgebra encapsulates the branching type of a given species of
systems as a \emph{functor}, for purposes of the present paper on the
category of sets. Such a functor $G\colon \Set\to\Set$ assigns to each
set~$X$ a set~$GX$, whose elements we think of as structured
collections over~$X$, and to each map $f\colon X\to Y$ a map $Gf\colon
GX\to GY$,
preserving identities and composition. E.g.\ the \emph{(covariant)
  powerset functor}~$\Pow$ assigns to each set~$X$ the powerset
$\Pow X$ of~$X$, and to each map $f\colon X\to Y$ the map
$\Pow f\colon \Pow X\to\Pow Y$ that takes direct images. (We mostly omit
the description of the action of functors on maps in the sequel.)
Systems with branching type described by~$G$ are then abstracted as
\emph{$G$-coalgebras}, i.e.\ pairs $(X,\gamma)$ consisting of a
set~$X$ of \emph{states} and a map $\gamma\colon X\to GX$, the
\emph{transition map}, which assigns to each state $x\in X$ a
structured collection $\gamma(x)$ of successors. For instance, a
$\Pow$-coalgebra assigns to each state a set of successors, and thus
is the same as a transition system.
\begin{example}\label{expl:coalg}
  \begin{longitemslist}
  \item Fix a set~$\Act$ of \emph{actions}. The functor
    $\Act\times(-)$ assigns to each set $X$ the set $\Act\times X$;
    composing this functor with the powerset functor, we obtain the
    functor $G=\Pow(\Act\times(-))$ whose coalgebras are precisely
    labelled transition systems (LTS): A $G$-coalgebra assigns to each
    state~$x$ a set of pairs $(\sigma,y)$, indicating that~$y$ is a
    successor of~$x$ under the action~$\sigma$.
  \item The \emph{(finite) distribution functor}~$\Dist$ maps a
    set~$X$ to the set of finitely supported discrete probability
    distributions on~$X$. These can be represented as probability mass
    functions $\mu\colon X\to[0,1]$, with $\sum_{x\in X}\mu(x)=1$ and with
    the \emph{support} $\{x\in X\mid \mu(x)>0\}$ being
    finite. Coalgebras for~$\Dist$ are precisely Markov
    chains. Composing with $\Act\times(-)$ as above, we obtain the
    functor $\Dist(\Act\times(-))$, whose coalgebras are
    \emph{generative probabilistic transition systems}, i.e.~assign
    to each state a distribution over pairs consisting of an action
    and a successor state. 
  \end{longitemslist}
\end{example}
As indicated in the introduction, we restrict attention to
\emph{finitary} functors~$G$, in which every element $t\in GX$ is
represented using only finitely many elements of~$X$; formally, each
set~$GX$ is the union of all sets $Gi_Y[GY]$ where $Y$ ranges over
finite subsets of~$X$ and $i_Y$ denotes the injection $i_Y\colon Y\into X$.
Concretely, this means that we restrict the set~$\Act$ of actions to
be finite, and work with the \emph{finite powerset functor}~$\Pfin$
(which maps a set~$X$ to the set of its finite subsets) in lieu
of~$\Pow$. ($\Dist$ as defined above is already finitary.)

Coalgebra comes with a natural notion of \emph{behavioural
  equivalence} of states. A \emph{morphism}
$f\colon (X,\gamma)\to(Y,\delta)$ of $G$-coalgebras is a map $f\colon X\to Y$ that
commutes with the transition maps, i.e.\
$\delta\cdot f=Gf\cdot\gamma$. Such a morphism is seen as preserving
the behaviour of states (that is, behaviour is defined as being
whatever is preserved under morphisms), and consequently states
$x\in X$, $z\in Z$ in coalgebras $(X,\gamma)$, $(Z,\zeta)$ are
\emph{behaviourally equivalent} if there exist coalgebra morphisms
$f\colon (X,\gamma)\to(Y,\delta)$, $g\colon (Z,\zeta)\to(Y,\delta)$ such that
$f(x)=g(z)$. For instance, states in LTSs are
behaviourally equivalent iff they are bisimilar in the standard sense,
and similarly, behavioural equivalence on generative probabilistic
transition systems coincides with the standard notion of probabilistic
bisimilarity~\cite{Klin09}. We have an alternative notion of
finite-depth behavioural equivalence: Given a $G$-coalgebra
$(X,\gamma)$, we define a series of maps $\gamma_n\colon X\to G^n1$
inductively by taking $\gamma_0$ to be the unique map $X\to 1$, and
$\gamma_{n+1} = G\gamma_n \cdot\gamma$. (These are the first $\omega$
steps of the \emph{canonical cone} from~$X$ into the \emph{final
  sequence} of~$G$~\cite{AdamekKoubek77}.) Then states $x,y$ in
coalgebras $(X,\gamma)$, $(Z,\zeta)$ are \emph{finite-depth
  behaviourally equivalent} if $\gamma_n(x)=\zeta_n(y)$ for all $n$;
in the case where~$G$ is finitary, finite-depth behavioural equivalence
coincides with behavioural equivalence~\cite{worrell}.

\section{Graded Monads and Graded Theories}

\noindent We proceed to recall background on system semantics via
graded monads introduced in our previous work~\cite{MiliusEA15}. We
formulate some of our results over general base categories~$\BC$,
using basic notions from category theory~\cite{MacLane98,Pierce91};
for the understanding of the examples, it will suffice to think of
$\BC=\Set$. Graded monads were originally introduced by
Smirnov~\cite{smirnov08} (with grades in a commutative monoid, which
we instantiate to the natural numbers):
\begin{defn}[Graded Monads]
  A \emph{graded monad}~$\BBM$ on a category $\cat C$ consists of a
  family of functors $(M_n\colon \cat C \to \cat C)_{n<\omega}$, a natural
  transformation $\unit\colon \Id \to M_0$ (the \emph{unit}) and a family
  of natural transformations 
  \iffull
  \[
    \mult^{nk}\colon M_n M_k \to M_{n+k}\quad (n,k<\omega)
  \]
  \else\/$\mult^{nk}\colon M_n M_k \to M_{n+k}$ for $n,k<\omega$, \fi
  (the \emph{multiplication}), satisfying the \emph{unit laws},
  $\mult^{0n}\cdot\unit M_n = \id_{M_n} = \mult^{n0}\cdot M_n\unit$,
  for all $n<\omega$, and the \emph{associative law}
  \iffull
  \[
    \begin{tikzcd}
    M_nM_kM_m\ar{d}{\mult^{nk}M_m}\ar{r}{M_n\mult^{km}} & M_nM_{k+m}\ar{d}{\mult^{n,k+m}} \\
      M_{n+k}M_m\ar{r}{\mult^{n+k,m}} & M_{n+k+m}
    \end{tikzcd}
    \qquad\text{for all $k,n,m<\omega$.}
  \]
  \else\/$\mult^{n,k+m} \cdot M_n \mult^{km} = \mult^{n+k,m} \cdot
  \mult^{nk}M_m$ for all $k,n,m<\omega$.\fi
\end{defn}
Note that it follows that $(M_0, \eta, \mu^{00})$ is a (plain)
monad. For $\cat C = \Set$, the standard equivalent presentation of
monads as algebraic theories carries over to graded monads. Whereas
for a monad $T$, the set $TX$ consists of terms over $X$ modulo
equations of the corresponding algebraic theory, for a graded monad
$(M_n)_{n<\omega}$, $M_nX$ consists of terms of uniform depth $n$
modulo equations:g

\begin{defn}[Graded Theories~\cite{MiliusEA15}]
  A \emph{graded theory} $(\Sigma,E,d)$ consists of an algebraic
  theory, i.e.\ a (possibly class-sized and infinitary) algebraic
  signature $\Sigma$ and a class $E$ of equations, and an assignment
  $d$ of a \emph{depth} $d(f)<\omega$ to every operation symbol
  $f\in\Sigma$. This induces a notion of a term \emph{having uniform
    depth $n$}: all variables have uniform depth $0$, and
  $f(t_1,\dots,t_n)$ with $d(f)=k$ has uniform depth $n+k$ if all
  $t_i$ have uniform depth $n$. (In particular, a constant $c$ has
  uniform depth $n$ for all $n\ge d(c)$). We require that all
  equations $t=s$ in $E$ have uniform depth, i.e.\ that both $t$ and
  $s$ have uniform depth~$n$ for some~$n$. Moreover, we require that
  for every set $X$ and every $k<\omega$, the class of terms of
  uniform depth $k$ over variables from $X$ modulo provable equality
  is small (i.e.\ in bijection with a set).
\end{defn}
\noindent Graded theories and graded monads on $\Set$ are essentially
equivalent concepts~\cite{smirnov08,MiliusEA15}. In particular, a
graded theory $(\Sigma,E,d)$ induces a graded monad~$\BBM$ by taking
$M_nX$ to be the set of $\Sigma$-terms over $X$ of uniform depth~$n$,
modulo equality derivable under $E$.

\begin{example}\label{E:graded-monad}
  We recall some examples of graded monads and theories~\cite{MiliusEA15}.
  \begin{longitemslist}
  \item\label{E:graded-monad:bisim} For every endofunctor $F$ on
    $\cat C$, the $n$-fold composition $M_n = F^n$ yields a graded
    monad with unit $\eta = \id_{\Id}$ and $\mu^{nk} = \id_{F^{n+k}}$.

  \item\label{E:graded-monad:kleisli} As indicated in the
    introduction, distributive laws yield graded monads: Suppose that
    we are given a monad $(T,\unit,\mult)$, an endofunctor $F$ on
    $\cat C$ and a distributive law of $F$ over $T$ (a so-called
    \emph{Kleisli law}), i.e.\ a natural transformation
    $\lambda\colon FT \to TF$ such that $\lambda \cdot F\eta = \eta F$
    and $\lambda \cdot F\mu = \mu F \cdot T\lambda \cdot \lambda T$.
    Define natural transformations $\lambda^n\colon F^nT \to TF^n$
    inductively by $\lambda^0 = \id_T$ and
    $\lambda^{n+1} = \lambda^{n}F \cdot F^{n}\lambda$. Then we obtain
    a graded monad with $M_n = TF^n$, unit $\eta$, and multiplication
    $\mu^{nk} = \mu F^{n+1} \cdot T\lambda^n F^k$. The situation is
    similar for distributive laws of~$T$ over~$F$ (so-called
    \emph{Eilenberg-Moore laws}).
  \item\label{E:graded-monad:tr} As a special case of
    \ref{E:graded-monad:kleisli}., for every monad $(T, \eta, \mu)$ on
    $\Set$ and every set $\Act$, we obtain a graded monad with
    $M_nX = T(\Act^n \times X)$.  Of particular interest to us will be
    the case where $T = \Pfin$, which is generated by the algebraic
    theory of join semilattices (with bottom). The arising graded monad
    $M_n=\Pfin(\Act^n\times X)$, which is
    associated with trace equivalence, is generated by the graded
    theory consisting, at depth~$0$, of the operations and equations
    of join semilattices, and additionally a unary operation of
    depth~$1$ for each $\sigma \in \Act$, subject to (depth-$1$)
    equations expressing that these unary operations distribute over
    joins.
  \end{longitemslist}
\end{example}

\myparagraph{Depth-1 Graded Monads and Theories}
where operations and equations have depth at most~$1$ are a particularly convenient case for
purposes of building algebras of graded monads; in the following, we elaborate on this
condition.
\begin{defn}[Depth-1 Graded
  Theory~\cite{MiliusEA15}]\label{D:d1}
  A graded theory is called \emph{depth-$1$} if all its operations
  and equations have depth at most~$1$. A graded monad on $\Set$ is
  \emph{depth-1} if it can be generated by a depth-1 graded theory.
\end{defn}
\begin{proposition}[Depth-1 Graded Monads~\cite{MiliusEA15}]\label{P:d1}
  A graded monad $((M_n),\eta,(\mu^{nk}))$ on $\Set$ is depth-$1$
  iff the diagram below is objectwise a coequalizer diagram in
  $\Set^{M_0}$ for all $n<\omega$:
  \begin{equation}\label{eq:mu1n}
    \xymatrix@1{M_1M_0M_n \ar@<3pt>[rr]^{M_1\mu^{0n}} 
      \ar@<-3pt>[rr]_{\mu^{10}M_n} && M_1 M_n \ar[r]^{\mu^{1n}} & M_{1+n}}.
  \end{equation}
\end{proposition}

\begin{example}\label{E:d1}
  All graded monads in \autoref{E:graded-monad} are depth $1$:
  for~\ref{E:graded-monad:bisim}., this is easy to see,
  for~\ref{E:graded-monad:tr}., it follows from the presentation as a
  graded theory, and for~\ref{E:graded-monad:kleisli}.,
  \iffull\/see~\hyperref[S:d1]{Appendix~\ref{S:d1}}.\else\/see~\cite{DorschEA19}.\fi
\end{example}

\noindent One may use the equivalent property of \autoref{P:d1} to
define depth-1 graded monads over arbitrary base
categories~\cite{MiliusEA15}. We show next that depth-1 graded monads
may be specified by giving only $M_0$, $M_1$, the unit~$\eta$, and $\mu^{nk}$
for $n+k \leq 1$.

\begin{theorem}
\label{thm:depth-1-graded-monads-M0}
Depth-$1$ graded monads are in bijective correspondence with
$6$-tuples $(M_0,M_1,\unit,\mult^{00},\mult^{10},\mult^{01})$ such
that the given data satisfy all applicable instances of the graded monad
laws.
\end{theorem}

\myparagraph{Semantics via Graded Monads} We next recall how graded
monads define \emph{graded semantics}:
\begin{defn}[Graded
  semantics~\cite{MiliusEA15}]\label{def:alpha-trace-semantics}
  Given a set functor~$G$, a \emph{graded semantics} for
  $G$-coalgebras consists of a graded monad
  $((M_n),\unit,(\mult^{nk}))$ and a natural transformation
  $\alpha\colon G\to M_1$.  The $\alpha$-\emph{pretrace sequence}
  $( \gamma^{(n)}\colon X\to M_nX )_{n<\omega}$ for a
  $G$-coalgebra $\gamma\colon X\to GX$ is defined by
  \[
    \gamma^{(0)} = (X \xrightarrow{\unit_X} M_0 X)
    \quad\text{and}\quad
    \gamma^{(n+1)} = (X
    \xrightarrow{\alpha_X\cdot\gamma} M_1 X \xrightarrow{M_1\gamma^{(n)}} M_1 M_n X
    \xrightarrow{\mult_X^{1n}} M_{n+1}X).
  \]
  The $\alpha$-\emph{trace sequence} $T^\alpha_\gamma$ is the sequence
  $( M_n!\cdot\gamma^{(n)}\colon X\to M_n1)_{n<\omega}$.
  
  In \Set, two states $x\in X$, $y\in Y$ of coalgebras
  $\gamma\colon X\to GX$ and $\delta\colon Y\to GY$ are $\alpha$-\emph{trace} (or
  \emph{graded}) \emph{equivalent} if
  $M_n!\cdot\gamma^{(n)}(x) = M_n!\cdot\delta^{(n)}(y)$ for all
  $n<\omega$.
\end{defn}
Intuitively, $M_nX$ consists of all length-$n$ \emph{pretraces}, i.e.\
traces paired with a poststate, and $M_n1$ consists of all length-$n$
traces, obtained by erasing the poststate. Thus, a graded semantics
extracts length-$1$ pretraces from successor structures. In the
following two examples we have $M_1 = G$; however, in general $M_1$
and $G$ can differ (\autoref{sec:ltbt}).

\begin{example}\label{expl:monads}
  Recall from \autoref{sec:prelim} that a $G$-coalgebra for the
  functor $G = \Pfin(\Act \times -)$ is just a finitely branching LTS.
  We recall two graded semantics that model the extreme ends of the
  linear time -- branching time spectrum~\cite{MiliusEA15}; more
  examples will be given in the next section

  \begin{longitemslist}
  \item \emph{Trace equivalence.} For $x,y \in X$ and $w\in \Act^*$,
    we write $x \xrightarrow{w} y$ if $y$ can be reached from $x$ on a
    path whose labels yield the word $w$, and
    $\T(x) = \{w \in \Act^* \mid \exists y \in X .\ x \xrightarrow{w}
    y\}$
    denotes the set of \emph{traces} of $x \in X$. States $x,y$ are
    \emph{trace equivalent} if $\T(x)=\T(y)$. To capture trace
    semantics of labelled transition systems we consider the graded
    monad with $M_nX = \Pow(\Act^n \times X)$ 
    (see \autoref{E:graded-monad}.\ref{E:graded-monad:tr}). The
    natural transformation $\alpha$ is the identity. For a
    $G$-coalgebra $(X, \gamma)$ and $x \in X$ we have that
    $\gamma^{(n)}(x)$ is the set of pairs $( w, y)$ with
    $w \in \Act^n$ and $x \xrightarrow{w} y$, i.e.\ pairs of
    length-$n$ traces and their corresponding poststate. Consequently,
    the~$n$-th component $M_n!  \cdot \gamma^{(n)}$ of the
    $\alpha$-trace sequence maps $x$ to the set of its length-$n$
    traces. Thus, $\alpha$-trace equivalence is standard trace
    equivalence~\cite{vanglabbeek2001linear}.

      Note that the equations presenting the graded monad $M_n$ in
      \autoref{E:graded-monad}.\ref{E:graded-monad:tr} bear a striking
      resemblance to the ones given by van Glabbeek to axiomatize
      trace equivalence of processes, with the difference that in his
      axiomatization actions do not distribute over the empty join. In
      fact, $a.0 = 0$ is clearly not valid for processes under trace
      equivalence. In the graded setting, this equation just expresses
      the fact that a trace which ends in a deadlock after $n$ steps
      cannot be extended to a trace of length $n+1$.

    \item\label{expl:monads:2} \emph{Bisimilarity.} By the discussion
      of the final sequence of a functor~$G$ (\autoref{sec:prelim}),
      the graded monad with $M_nX = G^nX$
      (\autoref{E:graded-monad}.\ref{E:graded-monad:bisim}), with
      $\alpha$ being the identity again, captures finite-depth
      behavioural equivalence, and hence behavioural equivalence
      when~$G$ is finitary. In particular, on finitely branching LTS,
      $\alpha$-trace equivalence is bisimilarity in this case.
  \end{longitemslist}
\end{example}

\section{A Spectrum of Graded Monads}\label{sec:ltbt}

We present graded monads for a range of equivalences on the linear
time -- branching time spectrum as well as probabilistic trace
equivalence for generative probabilistic systems (GPS), giving in each
case a graded theory and a description of the arising graded
monads. Some of our equations bear some similarity to van Glabbeek's
axioms for equality of process terms. There are also important
differences, however.  In particular, some of van Glabbeek's axioms
are implications, while ours are purely equational; moreover, van
Glabbeek's axioms sometimes nest actions, while we employ only
depth-$1$ equations (which precludes nesting of actions) in order to
enable the extraction of characteristic logics later. All graded
theories we introduce contain the theory of join semilattices, or in
the case of GPS convex algebras, whose operations are assigned
depth~$0$; we mention only the additional operations needed. We use
terminology introduced in \autoref{expl:monads}.

\myparagraph{Completed Trace
  Semantics}\label{ssec:completed-trace-semntics} refines trace
semantics by distinguishing whether traces can end in a deadlock.  We
define a depth-$1$ graded theory by extending the graded theory for
trace semantics (\autoref{E:graded-monad}) with a constant depth-$1$
operation~$\star$ denoting deadlock. The induced graded monad has
$M_0 X= \Pfin(X)$, $M_1 = \Pfin(\Act \times X + 1)$ (and
$M_nX=\Pfin(\Act^n\times X+\Act^{<n})$ where $\Act^{<n}$ denotes the
set of words over~$\Act$ of length less than~$n$).  The natural
transformation $\alpha_X\colon \Pfin(\Act\times X) \to M_1 X$ is given
by $\alpha(\emptyset)=\{\star\}$ and
$\alpha(S)=S\subseteq \Act \times X + 1$ for
$\emptyset\neq S\subseteq \Act\times X$.

\myparagraph{Readiness and Failures
  Semantics}\label{ssec:readiness_failures} refine completed trace
semantics by distinguishing which actions are available (readiness) or
unavailable (failures) after executing a trace.  Formally, given an
LTS, seen as a coalgebra $\gamma\colon X\to\Pfin(\Act\times X)$, we write
$I(x) = \Pfin\pi_1 \cdot \gamma(x) = \pi_1[ \gamma(x) ]$ ($\pi_1$ being
the first projection) for the set of actions available at~$x$, the
\emph{ready set} of~$x$.  A \emph{ready pair} of a state $x$ is a pair
$(w,A)\in\Act^*\times\Pfin(\Act)$ such that there exists~$z$ with
$x\overset w\to z$ and $A = I(z)$; a \emph{failure pair} is defined in
the same way except that $A\cap I(z)=\emptyset$.  Two states are
\emph{readiness (failures) equivalent} if they have the same ready
(failure) pairs.

We define a depth-$1$ graded theory by extending the graded theory for
trace semantics (\autoref{E:graded-monad}) with constant depth-$1$
operations $A$ for ready (failure) sets $A\subseteq\Act$. In case of
failures we add a monotonicity condition $A + A\cup B = A\cup B$ on
the constant operations for the failure sets. The resulting graded
monads both have $M_0X = \Pfin X$, and moreover
\( M_1X = \Pfin(\Act\times X+\Pfin\Act) \)
for readiness and \( M_1X = \Down(\Act\times X+\Pfin\Act) \)
for failures, where $\Down$ is down-closed finite powerset, w.r.t.\
the discrete order on $\Act\times X$ and set inclusion on $\Pfin\Act$.
The natural transformation
$\alpha_X\colon \Pfin(\Act\times X)\to M_1X$ is defined by
\( \alpha_X(S) = S \cup \{ \pi_1[S] \} \)
for readiness and
\( \alpha_X(S) = S \cup \{ A \subseteq\Act\mid A\cap\pi_1[S]=\emptyset
\} \) for failures semantics.

\myparagraph{Ready Trace and Failure Trace Semantics}\label{ssec:ready-failure-trace}
refine readiness and failures semantics, respectively, by distinguishing
which actions are available (ready trace) or unavailable (failure trace) at each step
of the trace. Formally, a \emph{ready trace} of a state~$x$ is a sequence
$A_0a_1A_1\ldots a_n A_n\in(\Pfin\Act\times\Act)^\ast\times\Pfin\Act$
such that there exist transitions
$x=x_0\overset{a_1}\to x_1\ldots \overset{a_n}\to x_n$ where each
$x_i$ has ready set $I(x_i)=A_i$. A \emph{failure trace} has the same
shape but we require that each $A_i$ is a \emph{failure set} of $x_i$,
i.e.\ $I(x_i)\cap A_i=\emptyset$. States are \emph{ready (failure) trace equivalent} if they have the
same ready (failure) traces.

For ready traces, we define a depth-$1$ graded theory with depth-$1$
operations $\langle A,\sigma \rangle$ for $\sigma\in\Act$,
$A\subseteq\Act$ and a depth-$1$ constant $\star$, denoting deadlock,
and equations
$ \langle A,\sigma \rangle(\sum_{j\in J} x_j) = \sum_{j\in J}\langle
A,\sigma \rangle (x_j)$.
The resulting graded monad is simply the graded monad capturing
completed trace semantics for labelled transition systems where the
set of actions is changed from $\Act$ to $\Pfin\Act\times\Act$. For
failure traces, we additionally impose the equation
$ \langle A,\sigma \rangle(x)+\langle A\cup B,\sigma
\rangle(x)=\langle A\cup B,\sigma \rangle(x)$,
which in the set-based description of the graded monad corresponds to
downward closure of failure sets.

The resulting graded monads both have $M_0X =\Pfin X$; for ready
traces, \( M_1X = \Pfin(( \Pfin\Act\times\Act ) \times X + 1) \)
and for failure traces,
\( M_1X = \Down(( \Pfin\Act\times\Act ) \times X + 1) \),
where $\Down$ is down-closed finite powerset, w.r.t.\ the order
imposed by the above equation.

For ready trace semantics we define the natural transformation
$\alpha_X\colon \Pfin(\Act\times X) \to M_1X$ by
$\alpha_X(\emptyset) = \{\star\}$ and
$\alpha_X(S) = \{ (( \pi_1[S], \sigma),x) \mid ( \sigma ,x
) \in S \})$ for $S \neq\emptyset$.
For failure traces we define
$\alpha_X(\emptyset) = \{\star\}$ and
$\alpha(S) = \{(( A, \sigma),x) \mid
(\sigma,x)\in S, A\cap \pi_1[S]=\emptyset\}$
for $S\neq\emptyset$; note that in the latter case, $\alpha(S)$ is
closed under decreasing failure sets.

\myparagraph{Simulation Equivalence}\label{ssec:simulation}
declares
two states to be equivalent if they simulate each other in the
standard sense.  We define a depth-$1$ graded theory with the same
signature as for trace equivalence but instead of join preservation
require only that each $\sigma$ is monotone, i.e.\
$\sigma(x + y) + \sigma(x) = \sigma(x + y)$. The arising graded
monad~$M_n$ is equivalently described as follows. We define the sets
$M_nX$ inductively, along with an ordering on $M_nX$. We take
$M_0X=\Pfin X$, ordered by set inclusion. We then order the elements
of $\Act\times M_nX$ by the product ordering of the discrete order
on~$\Act$ and the given ordering on~$M_nX$, and take $M_{n+1}X$ to be
the set of downclosed subsets of $\Act\times M_nX$, denoted
$\Down(\Act\times M_nX)$, ordered by set inclusion.  The natural
transformation
$\alpha_X\colon \Pow(\Act\times X) \to \Down(\Act\times\Pfin(X))$ is defined
by $\alpha_X(S)={\downarrow}\{(s,\{x\})\mid (s,x)\in S\}$, where
$\downarrow$ denotes downclosure.

\myparagraph{Ready Simulation Equivalence}
refines simulation equivalence by requiring
additionally that related states have the same ready set.
States $x$ and $y$ are \emph{ready similar} if they are related by
some ready simulation, and ready simulation equivalent if there are
mutually ready similar.
The depth-$1$ graded theory combines the signature for ready traces
with the equations for simulation, i.e.\ only requires the operations
$\langle A,\sigma \rangle$ to be monotone.

\myparagraph{Probabilistic Trace Equivalence}\label{ssec:GPS} is the
standard trace semantics for generative probabilistic systems (GPS),
equivalently, coalgebras for the functor \( \dist(\Act\times\Id) \)
where $\dist$ is the monad of finitary distributions
(\autoref{expl:coalg}).
Probabilistic trace equivalence is captured by the graded monad
$M_nX = \dist (\Act^n\times X)$, as described in
\autoref{E:graded-monad}.\ref{E:graded-monad:kleisli}. The
corresponding graded theory arises by replacing the join-semilattice
structure featuring in the above graded theory for trace equivalence
by the one of \emph{convex algebras}, i.e.\ the algebras for the
monad~$\dist$. Recall~\cite{Doberkat06,Doberkat08} that a convex
algebra is a set $X$ equipped with finite convex sum operations: For
every $p \in [0,1]$ there is a binary operation $\boxplus_p$ on $X$,
and these operations satisfy the equations
\(
x \boxplus_p x = x,
x \boxplus_p y = y \boxplus_{1-p} x,
x \boxplus_0 y = y,
x \boxplus_p (y \boxplus_q z) = (x
\boxplus_{p/r} y) \boxplus_r z,
\)
where $p, q \in [0,1]$, $x, y, z \in X$, and $r = (p + (1-p)q) > 0$
(i.e.\ $p+q>0$) in the last equation~\cite{Jacobs10}.  Again, we have
depth-$1$ operations $\sigma$ for action $\sigma\in\Act$, now
satisfying the equations
\begin{math}
  \sigma(x \boxplus_p y) = \sigma(x) \boxplus_p \sigma(y).
\end{math}




\section{Graded Logics}\label{sec:logics}

\noindent Our next goal is to extract \emph{characteristic logics}
from graded monads in a systematic way, with \emph{characterizing}
meaning that states are logically indistinguishable iff they are
equivalent under the semantics at hand. We will refer to these logics
as \emph{graded logics}; the implication from graded equivalence to
logical indistinguishability is called \emph{invariance}, and the
converse implication \emph{expressiveness}. E.g.\ standard modal logic
with the full set of Boolean connectives is invariant under
bisimilarity, and the corresponding expressiveness result is known as
the \emph{Hennessy-Milner theorem}. This result has been lifted to
coalgebraic generality early on, giving rise to the \emph{coalgebraic
  Hennessy-Milner theorem}~\cite{Pattinson04,Schroder08}. In previous
work~\cite{MiliusEA15}, we have related graded semantics to modal
logics extracted from the graded monad in the envisaged fashion. These
logics are invariant by construction; the main new result we
present here is a generic \emph{expressiveness} criterion, to be
discussed in \autoref{sec:expr}. The key ingredient in this criterion
are \emph{canonical} graded algebras, which we newly introduce here,
providing a recursive-evaluation style reformulation of the semantics
of graded logics.

A further key issue in characteristic modal logics is the choice of
propositional operators; e.g.\ notice that when $\trdiamond{\sigma}$
denotes the usual Hennessy-Milner style diamond operator for an
action~$\sigma$, the formula
$\trdiamond{\sigma}\top\land\trdiamond{\tau}\top$ is invariant under
trace equivalence (i.e.~the corresponding property is closed under
under trace equivalence) but the formula
$\trdiamond{\sigma}(\trdiamond{\sigma}\top\land\trdiamond{\tau}\top)$,
built from the former by simply prefixing with~$\trdiamond{\sigma}$,
is not, the problem being precisely the use of conjunction. While in
our original setup, propositional operators were kept implicit, that
is, incorporated into the set of modalities, we provide an explicit
treatment of propositional operators in the present paper. Besides
adding transparency to the syntax and semantics, having first-class
propositional operators will be a prerequisite for the formulation of
the expressiveness theorem.

\myparagraph{Coalgebraic Modal Logic} To provide context, we briefly
recall the setup of \emph{coalgebraic modal
  logic}~\cite{Pattinson04,Schroder08}. Let~$2$ denote the
set~$\{\bot,\top\}$ of Boolean truth values; we think of the set~$2^X$
of maps $X\to 2$ as the set of predicates on~$X$. Coalgebraic logic in
general abstracts systems as coalgebras for a functor~$G$, like we do
here; fixes a set~$\Lambda$ of \emph{modalities} (unary for the sake
of readability); and then interprets a modality $L\in\Lambda$ by the
choice of a \emph{predicate lifting}, i.e.\ a natural transformation
\begin{equation*}
  \Sem{L}_X\colon 2^X\to 2^{GX}.
\end{equation*}
By the Yoneda lemma, such natural transformations are in bijective
correspondence with maps $G2\to 2$~\cite{Schroder08}, which we shall
also denote as $\Sem{L}$. In the latter formulation, the recursive
clause defining the interpretation $\Sem{L\phi}\colon X\to 2$, for a
modal formula~$\phi$, as a state predicate in a $G$-coalgebra
$\gamma \colon X\to GX$ is then
\begin{equation}\label{eq:coalg-modality}
  \Sem{L\phi}= (X\xrightarrow{\gamma}GX\xrightarrow{ G\Sem{\phi}} G2\xrightarrow{\Sem{L}}2).
\end{equation}
E.g.\ taking $G=\Pfin(\Act\times-)$ (for labelled transition systems),
we obtain the standard semantics of the Hennessy-Milner diamond
modality $\trdiamond{\sigma}$ for~$\sigma\in\Act$ via the predicate
lifting
\begin{equation*}
  \Sem{\trdiamond{\sigma}}_X(f)=\{B\in\Pfin(\Act\times X)\mid
  \exists x.\,(\sigma,x)\in B\land f(x)=\top\}\qquad(\text{for
    $f\colon  X\to 2$}).
\end{equation*}
It is easy to see that \emph{coalgebraic modal logic}, which combines
coalgebraic modalities with the full set of Boolean connectives, is
invariant under finite-depth behavioural equivalence
(\autoref{sec:prelim}). Generalizing the classical Hennessy-Milner
theorem~\cite{HennessyMilner85}, the \emph{coalgebraic Hennessy-Milner
  theorem}~\cite{Pattinson04,Schroder08} shows that conversely,
coalgebraic modal logic \emph{characterizes} behavioural equivalence,
i.e.\ logical indistinguishability implies behavioural equivalence,
provided that~$G$ is finitary (implying coincidence of behavioural
equivalence and finite-depth behavioural equivalence) and~$\Lambda$ is
\emph{separating}, i.e.\ for every finite set~$X$, the set
\begin{equation*}
  \Lambda(2^X)=\{\Sem{L}(f)\mid f\in 2^X\}
\end{equation*}
of maps $GX\to 2$ is jointly injective.

We proceed to introduce the syntax and semantics of graded logics.
\myparagraph{Syntax} We parametrize the syntax of \emph{graded logics}
over
\begin{itemize}
\item a set~$\Theta$ of \emph{truth constants},
\item a set~$\PropOps$ of \emph{propositional operators} with assigned
  finite arities, and
\item a set~$\Lambda$ of \emph{modalities} with assigned arities.
\end{itemize}
For readability, we will restrict the technical exposition to unary
modalities; the treatment of higher arities requires no more than
additional indexing (and we will use $0$-ary modalities in the
examples). E.g.\ standard Hennessy-Milner logic is given by
$\Lambda=\{\trdiamond{\sigma}\mid \sigma\in\Act\}$ and~$\PropOps$
containing all Boolean connectives. Other logics will be determined by
additional or different modalities, and often by fewer propositional
operators. Formulae of the logic are restricted to have uniform depth,
where propositional operators have depth~$0$ and modalities have
depth~$1$; a somewhat particular feature is that truth constants can
have top-level occurrences only in depth-$0$ formulae.  That is,
formulae~$\phi,\phi_1,\dots$ of depth~$0$ are given by the grammar
\begin{equation*}
  \phi\Coloneqq p(\phi_1,\dots,\phi_k) \mid c
    \qquad (p\in\PropOps\text{ $k$-ary}, c\in\Theta),
\end{equation*}
and formulae~$\phi$ of depth $n+1$ by
\begin{equation*}
  \phi\Coloneqq p(\phi_1,\dots,\phi_k) \mid L\psi
    \qquad (p\in\PropOps\text{ $k$-ary}, L\in\Lambda)
\end{equation*}
where $\phi_1,\dots,\phi_n$ range over formulae of depth $n+1$ and
$\psi$ over formulae of depth~$n$. 

\myparagraph{Semantics} The semantics of graded logics is parametrized
over the choice of \emph{a functor~$G$, a depth-$1$ graded monad
  $\BBM=((M_n)_{n<\omega},\eta,$ $(\mu^{nk})_{n,k<\omega})$, and a
  graded semantics~$\alpha\colon G\to M_1$, which we fix for the
  remainder of the paper}. It was originally given by translating
formulae into \emph{graded algebras} and then defining formula
evaluation by the universal property of $(M_n1)$ as a free graded
algebra~\cite{MiliusEA15}; here, we reformulate the semantics in a
more standard style by recursive clauses, using canonical graded
algebras. In general, the notion of graded algebra is defined as
follows~\cite{MiliusEA15}.
\begin{defn}[Graded algebras]
  Let $n<\omega$. A \emph{(graded) $M_n$-algebra}
  $A=((A_k)_{k\le n},(a^{mk})_{m+k\le n})$ consists of carrier
  sets~$A_k$ and structure maps
  \begin{equation*}
    a^{mk}\colon M_mA_k\to A_{m+k}
  \end{equation*}
  satisfying the laws 
  \begin{equation}\label{diag:alg}
    \begin{tikzcd}
      A_k \arrow{r}{\eta_{A_k}} \arrow[equal]{dr}
       & M_0 A_k\arrow{d}{a^{0k}} & 
      M_m M_r A_k \arrow{r}{M_m a^{rk}}
        \arrow{d}[left]{\mu^{mr}_{A_k}} &
        M_m A_{r+k} \arrow{d}{a^{m,r+k}} \\
        & A_k &  M_{m+r}A_k \arrow{r}{a^{m+r,k}} & A_{m+r+k}
    \end{tikzcd}
  \end{equation}
  for all $k\le n$ (left) and all $m,r,k$ such that $m+r+k\le n$
  (right), respectively. An \emph{$M_n$-morphism}~$f$ from~$A$ to an
  $M_n$-algebra $B=((B_k)_{k\le n},(b^{mk})_{m+k\le n})$ consists of
  maps $f_k\colon A_k\to B_k$, $k\le n$, such that
  $f_{m+k}\cdot a^{mk}=b^{mk}\cdot M_mf_k$
  for all $m,k$ such that $m+k\le n$.
\end{defn}
\noindent 
We view the carrier~$A_k$ of an~$M_n$-algebra as the set of algebra
elements that have already absorbed operations up to depth~$k$. As in
the case of plain monads, we can equivalently describe graded algebras
in terms of graded theories: If $\BBM$ is generated by a graded theory
$\GTh=(\Sigma,E,d)$, then an $M_n$-algebra interprets each operation
$f\in\Sigma$ of arity~$r$ and depth~$d(f)=m$ by maps
$f^A_k\colon A_k^r\to A_{m+k}$ for all $k$ such that $m+k\le n$; this
gives rise to an inductively defined interpretation of terms
(specifically, given a valuation of variables in~$A_m$, terms of
uniform depth~$k$ receive values in~$A_{k+m}$, for $k+m\le n$), and
subsequently to the expected notion of satisfaction of equations.


While in general, graded algebras are monolithic objects, for
\mbox{depth-$1$} graded monads we can construct them in a modular
fashion from $M_1$-algebras~\cite{MiliusEA15}; we thus restrict
attention to $M_0$- and $M_1$-algebras in the following. We note that
an $M_0$-algebra is just an Eilenberg-Moore algebra for the
monad~$M_0$. An $M_1$-Algebra~$A$ consists of $M_0$-algebras
$(A_0,a^{00}\colon M_0A_0\to A_0)$ and $(A_1,a^{01}\colon M_0A_1\to A_1)$, and a
\emph{main structure map} $a^{10}\colon M_1A_0\to A_1$ satisfying two
instances of the right-hand diagram in~\eqref{diag:alg}, one of which
says that $a^{10}$ is a morphism of $M_0$-algebras
(\emph{homomorphy}), and the other that the diagram
\begin{equation}
  \label{diag:algebra-coeq}
  \begin{tikzcd}[column sep=large]
    M_1M_0A_0 \arrow[shift left]{r}[above]{\mu^{10}}
    \arrow[shift right]{r}[below]{M_1a^{00}}& M_1A_0 \arrow{r}{a^{10}} & A_1,
  \end{tikzcd}
\end{equation}
which by the laws of graded monads consists of $M_0$-algebra
morphisms, commutes (\emph{coequalization}). We will often refer to an
$M_1$-algebra by just its main structure map.

We will use $M_1$-algebras as interpretations of the modalities in
graded logics, generalizing the previously recalled interpretation of
modalities as maps $G2\to 2$ in branching-time coalgebraic modal
logic. We fix an $M_0$-algebra $\Omega$ of \emph{truth values}, with
structure map $\truthstruct\colon M_0\Omega\to\Omega$ (e.g.\
for~$G=\Pfin$, $\Omega$ is a join semilattice).  Powers~$\Omega^n$
of~$\Omega$ are again
$M_0$-algebras. 
A modality $L\in\Lambda$ is interpreted as an $M_1$-algebra
$A=\Sem{L}$ with carriers $A_0=A_1=\Omega$ and
$a^{01}=a^{00}=\truthstruct$. Such an $M_1$-algebra is thus specified
by its main structure map $a^{10}\colon M_1\Omega\to\Omega$ alone, so
following the convention indicated above we often write $\Sem{L}$ for
just this map. 
The evaluation of modalities is defined using canonical
$M_1$-algebras:
\begin{defn}[Canonical algebras]
  The \emph{$0$-part} of an~$M_1$-algebra~$A$ is the $M_0$-algebra
  $(A_0,a^{00})$. Taking $0$-parts defines a functor $U_0$ from
  $M_1$-algebras to $M_0$-algebras. An $M_1$-algebra is
  \emph{canonical} if it is free, w.r.t.\ $U_0$, over its
  $0$-part. For~$A$ canonical and a modality $L\in\Lambda$, we
  denote the unique morphism $A_1\to\Omega$ extending an
  $M_0$-morphism $f\colon A_0\to\Omega$ to an $M_1$-morphism $A\to\Sem{L}$
  by~$\Sem{L}(f)$, i.e.\ $\Sem{L}(f)$ is the unique $M_0$-morphism
  such that
  \iffull
  the square below commutes:
  \begin{equation}\label{diag:L(f)}
    \begin{tikzcd}
      M_1 A_0 \arrow{r}{M_1f} \arrow{d}[left]{a^{10}} 
      & M_1\Omega\arrow{d}{\Sem{L}}\\
      A_1 \arrow{r}[below]{\Sem{L}(f)} & \Omega
    \end{tikzcd}
  \end{equation}
  \else
  the following equation holds:
  \begin{equation}\label{diag:L(f)}
   (M_1 A_0 \xrightarrow{M_1 f} M_1\Omega \xrightarrow{\Sem{L}} \Omega)
   =
   (M_1 A_0 \xrightarrow{a^{10}} A_1 \xrightarrow{\Sem{L}(f)} \Omega).
 \end{equation}
 \fi
\end{defn}
\begin{lemma}\label{lem:canonical}
  An $M_1$-algebra~$A$ is canonical iff \eqref{diag:algebra-coeq} is a
  (reflexive) coequalizer diagram in the category of $M_0$-algebras.
\end{lemma}
\noindent By the above lemma, we obtain a key example of canonical
$M_1$-algebras:
\begin{corollary}
  If $\BBM$ is a depth-$1$ graded monad, then for every~$n$ and every
  set~$X$, the $M_1$-algebra with carriers $M_nX,M_{n+1}X$ and
  multiplication as algebra structure is canonical.
\end{corollary}


\noindent Further, we interpret truth constants $c\in\Theta$ as
elements of~$\Omega$, understood as maps $\hat{c}\colon 1\to\Omega$,
and $k$-ary propositional operators $p\in\PropOps$ as $M_0$-homomorphisms
$
\Sem{p}\colon\Omega^k\to\Omega.
$
In our examples on the linear time -- branching time spectrum,~$M_0$
is either the identity or, most of the time, the finite powerset
monad. In the former case, all truth functions are $M_0$-morphisms. In
the latter case, the $M_0$-morphisms $\Omega^k\to \Omega$ are the
join-continuous functions; in the standard case where $\Omega=2$ is
the set of Boolean truth values, such functions~$f$ have the form
$f(x_1,\dots,x_k)=x_{i_1}\lor\dots\lor x_{i_l}$, where
$i_1,\dots,i_l\in\{1,\dots,k\}$. We will see one case where $M_0$ is
the distribution monad; then $M_0$-morphisms are affine
maps.

The semantics of a formula~$\phi$ in graded logic is defined recursively
as an $M_0$-morphism
$\Sem{\phi}\colon (M_n1, \mu^{0n}_1) \to (\Omega, o)$ by
\begin{equation*}
  \Sem{c}  = (M_01\xrightarrow{M_0\hat c}M_0\Omega \xrightarrow{o}\Omega)\quad
  \Sem{p(\phi_1,\dots,\phi_k)}  =\Sem{p}\cdot\langle\Sem{\phi_1},\dots,
                                 \Sem{\phi_k}\rangle\quad
  \Sem{L\phi}  = \Sem{L}(\Sem{\phi}).
\end{equation*}
The evaluation of~$\phi$ in a coalgebra $\gamma\colon X\to GX$ is then given
by composing with the trace sequence, i.e.\ as
\begin{equation}\label{eq:formula-eval}
  X\xrightarrow{M_n!\cdot\gamma^{(n)}} M_n1\xrightarrow{\Sem\phi}\Omega.
\end{equation}
In particular, graded logics are, by construction, invariant under the
graded semantics. 

\begin{example}[Graded logics]\label{expl:logics}
  We recall the two most basic examples, fixing $\Omega=2$ in both
  cases, and $\top$ as the only truth constant:
  \begin{longitemslist}
  \item \emph{Finite-depth behavioural equivalence:} Recall that the
    graded monad $M_nX=G^nX$ captures finite-depth behavioural
    equivalence on $G$-coalgebras. Since~$M_0$ is the identity monad,
    $M_0$-algebras are just sets. Thus, every function $2^k\to 2$ is
    an $M_0$-morphism, 
    so we can use all Boolean operators as propositional
    operators. Moreover, $M_1$-algebras are just maps
    $a^{10}\colon GA_0\to A_1$. Such an $M_1$-algebra is canonical iff
    $a^{10}$ is an isomorphism, and modalities are interpreted as
    $M_1$-algebras $G2\to 2$, with the evaluation according
    to~\eqref{diag:L(f)} and~\eqref{eq:formula-eval} corresponding
    precisely to the semantics of modalities in coalgebraic
    logic~\eqref{eq:coalg-modality}. Summing up, we obtain precisely
    coalgebraic modal logic as summarized above in this case. In our
    running example $G=\Pfin(\Act\times(-))$, we take modalities
    $\Diamond_\sigma$ as above, with
    $\Sem{\Diamond_\sigma}\colon\Pfin(\Act\times 2)\to 2$ defined by
    $\Sem{\Diamond_\sigma}(S)=\top$ iff $(\sigma,\top)\in S$,
    obtaining precisely classical Hennessy-Milner
    logic~\cite{HennessyMilner85}.
  \item \emph{Trace equivalence:} Recall that the trace semantics of
    labelled transition systems with actions in~$\Act$ is modelled by
    the graded monad $M_nX=\Pfin(\Act^n\times X)$. As indicated above,
    in this case we can use disjunction as a propositional operator
    since $M_0=\Pfin$. Since the graded theory for $M_n$ specifies for
    each $\sigma\in\Act$ a unary depth-$1$ operation that distributes
    over joins, we find that the maps $\Sem{\Diamond_\sigma}$ from the
    previous example (unlike their duals $\Box_\sigma$) induce
    $M_1$-algebras also in this case, so we obtain a graded trace
    logic featuring precisely diamonds and disjunction, as expected.
  \end{longitemslist}
  We defer the discussion of further examples, including ones where
  $\Omega=[0,1]$, to the next section, where we will simultaneously
  illustrate the generic expressiveness result
  (\autoref{expl:ltbt-logics}). 
\end{example}

\begin{remark}
  One important class of examples where the above approach to
  characteristic logics will \emph{not} work without substantial
  further development are simulation-like equivalences, whose
  characteristic logics need
  conjunction~\cite{vanglabbeek2001linear}. Conjunction is not an
  $M_0$-morphism for the corresponding graded monads identified in
  \autoref{sec:ltbt}, which both have $M_0=\Pfin$. A related and maybe
  more fundamental observation is that formula evaluation is not
  $M_0$-morphic in the presence of conjunction; e.g.\ over simulation
  equivalence, the evaluation map
  $M_11=\Down(\Act\times\Pfin(1))\to 2$ of the formula
  $\trdiamond{\sigma}\top\land\trdiamond{\tau}\top$ fails to be
  join-continuous for distinct $\sigma,\tau\in\Act$. We leave the
  extension of our logical framework to such cases to future work,
  expecting a solution in elaborating the theory of graded monads,
  theories, and algebras over the category of partially ordered sets,
  where simulations live more naturally (e.g.~\cite{KapulkinEA12}).
\end{remark}



%
\section{Expressiveness}\label{sec:expr}

We now present our main result, an expressiveness criterion for graded
logics, which states that a graded logic characterizes the given
graded semantics if it has enough modalities propositional operators,
and truth constants. Both the criterion and its proof now fall into
place naturally and easily, owing to the groundwork laid in the
previous section, in particular the reformulation of the semantics in
terms of canonical algebras:
\begin{defn}\label{def:separation}
  We say that a graded logic with set~$\Omega$ of truth values and
  sets~$\Theta$,~$\PropOps$,~$\Lambda$ of truth constants,
  propositional operators, and modalities, respectively, is
  \begin{longitemslist}
  \item \emph{depth-$0$ separating} if the family of maps
    $\Sem{c}\colon M_01\to\Omega$, for truth constants
    $c\in\Theta$, is jointly injective; and
  \item \emph{depth-$1$ separating} if, whenever $A$ is a canonical
    $M_1$-algebra and $\FA$ is a jointly injective set of
    $M_0$-homomorphisms $A_0\to\Omega$ that is closed under the
    propositional operators in~$\PropOps$ (in the sense that
    $\Sem{p}\cdot \langle f_1,\dots,f_k\rangle\in\FA$ for
    $f_1,\dots,f_k\in\FA$ and $k$-ary $p\in\PropOps$), then the
    set 
    \[
      \Lambda(\FA)\coloneqq\{\Sem{L}(f)\colon A_1\to\Omega\mid L\in\Lambda,f\in\FA\}.
    \]
    of maps is jointly injective.
  \end{longitemslist}
\end{defn}
\begin{theorem}[Expressiveness]\label{thm:expr}
  If a graded logic is both depth-$0$ separating and depth-$1$
  separating, then it is expressive.
\end{theorem}

\begin{example}[Logics for bisimilarity]\label{expl:bisim-logics}
  We note first that the existing coalgebraic Hennessy-Milner theorem,
  for branching time equivalences and coalgebraic modal logic with
  full Boolean base over a finitary
  functor~$G$~\cite{Pattinson04,Schroder08}, as recalled in
  Section~\ref{sec:logics}, is a special case of \autoref{thm:expr}:
  We have already seen in \autoref{expl:logics} that coalgebraic modal
  logic in the above sense is an instance of our framework for the
  graded monad $M_nX=G^nX$. Since $M_0=\id$ in this case, depth-$0$
  separation is vacuous. As indicated in \autoref{expl:logics},
  canonical $M_1$-algebras are w.l.o.g.\ of the form $\id\colon GX\to GX$,
  where for purposes of proving depth-$1$ separation, we can restrict
  to finite~$X$ since~$G$ is finitary. Then, a set~$\FA$ as in
  \autoref{def:separation} is already the whole powerset $2^X$, so
  depth-$1$ separation is exactly the previous notion of separation.
 
  A well-known particular case is probabilistic bisimilarity on Markov
  chains, for which an expressive logic needs only probabilistic
  modalities $\Diamond_p$ `with probability at least~$p$' and
  conjunction~\cite{DesharnaisEA98}. This result (later extended to
  more complex composite functors~\cite{MossViglizzo06}) is also
  easily recovered as an instance of \autoref{thm:expr}, using the
  same standard lemma from measure theory as in \emph{op.~cit.},
  which states that measures are uniquely determined by their values
  on a generating set of the underlying $\sigma$-algebra that is
  closed under finite intersections (corresponding to the set~$\FA$
  from \autoref{def:separation} being closed under conjunction).
\end{example}

\begin{remark}
  For behavioural equivalence, i.e.\ $M_nX=G^nX$ as in the above
  example, the inductive proof of our expressiveness theorem
  essentially instantiates to Pattinson's proof of the coalgebraic
  Hennessy-Milner theorem by induction over the terminal
  sequence~\cite{Pattinson04}. One should note that although the
  coalgebraic Hennessy-Milner theorem can be shown to hold for larger
  cardinal bounds on the branching by means of a direct quotienting
  construction~\cite{Schroder08}, the terminal sequence argument goes
  beyond finite branching only in corner cases.
\end{remark}

\begin{example}[Expressive graded logics on the linear time -- branching time spectrum]\label{expl:ltbt-logics}
  We next extract graded logics from some of the graded monads
  for the linear time -- branching time spectrum introduced in
  \autoref{sec:ltbt}, and show how in each case, expressiveness is an
  instance of \autoref{thm:expr}. Bisimilarity is already covered by
  the previous example. Depth-$0$ separation is almost always trivial
  and not mentioned further. Unless mentioned otherwise, all logics
  have disjunction, enabled by $M_0$ being powerset as discussed in
  the previous section. Most of the time, the logics are essentially
  already given by van Glabbeek (with the exception that we show that
  one can add disjunction)~\cite{vanglabbeek2001linear}; the emphasis
  is entirely on uniformization.
  \begin{longitemslist}
  \item \emph{Trace equivalence:} As seen in \autoref{expl:logics},
    the graded logic for trace equivalence features (disjunction and)
    diamond modalities $\Diamond_\sigma$ indexed over actions
    $\sigma\in\Act$. The ensuing proof of depth-$1$ separation uses
    canonicity of a given $M_1$-algebra~$A$ only to obtain that the
    structure map $a^{10}$ is surjective. The other key point is that
    a jointly injective collection~$\FA$ of $M_0$-homomorphisms
    $A_0\to 2$, i.e.\ join preserving maps, has the stronger
    separation property that whenever $x\not\le y$ then there exists
    $f\in\FA$ such that $f(x)=\top$ and $f(y)=\bot$.
  \item Graded logics for completed traces, readiness, failures, ready
    traces, and failure traces are developed from the above by adding
    constants or additionally indexing modalities over sets of
    actions, with only little change to the proofs of depth-$1$
    separation. For completed trace equivalence, we just add a $0$-ary
    modality $\star$ indicating deadlock. For ready trace equivalence,
    we index the diamond modalities $\Diamond_\sigma$ with sets
    $I\subseteq\Act$; formulae $\Diamond_{\sigma,I}\phi$ are then read
    `the current ready set is~$I$, and there is a $\sigma$-successor
    satisfying~$\phi$'. For failure trace equivalence we proceed in
    the same way but read the index~$I$ as `$I$ is a failure set at
    the current state'. For readiness equivalence and failures
    equivalence, we keep the modalities~$\Diamond_\sigma$ unchanged
    from trace equivalence and instead introduce $0$-ary
    modalities~$r_I$ indicating that~$I$ is the ready set or a failure
    set, respectively, at the current state, thus ensuring that
    formulae do not continue after postulating a ready set.
  \end{longitemslist}
\end{example}

\begin{example}[Probabilistic traces]\label{expl:prob-trace}
  We have recalled in \autoref{sec:ltbt} that probabilistic trace
  equivalence of generative probabilistic transition systems can be
  captured as a graded semantics using the graded
  monad~$M_nX=\Dist(\Act^n\times X)$, with $M_0$-algebras being convex
  algebras. In earlier work~\cite{MiliusEA15} we have noted that a
  logic over the set $\Omega=[0,1]$ of truth values (with the usual
  convex algebra structure) featuring rational truth constants, affine
  combinations as propositional operators (as indicated in
  \autoref{sec:logics}), and modal operators $\langle\sigma\rangle$,
  interpreted by $M_1$-algebras
  $\Sem{\langle\sigma\rangle}\colon M_1[0,1]\to[0,1]$ defined by
  \iffull
  \begin{equation*}\textstyle
    \Sem{\langle\sigma\rangle}(\mu)=\sum_{r\in [0,1]}r\mu(\sigma,r)
  \end{equation*}
  \else\/$\Sem{\langle\sigma\rangle}(\mu)=\sum_{r\in
    [0,1]}r\mu(\sigma,r)$ \fi
  is invariant under probabilistic trace equivalence. By our
  expressiveness criterion, we recover the result that this logic
  is expressive for probabilistic trace semantics
  (see e.g.~\cite{bernardo-botta:characterising-logics}).
\end{example}
\section{Conclusion and Future Work}

We have provided graded monads modelling a range of process
equivalences on the linear time -- branching time spectrum, presented
in terms of carefully designed graded algebraic theories. From these
graded monads, we have extracted characteristic modal logics for the
respective equivalences systematically, following a paradigm of graded
logics that grows out of a natural notion of graded algebra. Our main
technical results concern the further development of the general
framework for graded logics; in particular, we have introduced a
first-class notion of propositional operator, and we have established
a criterion for \emph{expressiveness} of graded logics that
simultaneously takes into account the expressive power of the
modalities and that of the propositional base. (An open question that
remains is whether an expressive logic always exists, as it does in
the branching-time setting~\cite{Schroder08}.) Instances of this
result include, for instance, the coalgebraic Hennessy-Milner
theorem~\cite{Pattinson04,Schroder08}, Desharnais et al.'s
expressiveness result for probabilistic modal logic with only
conjunction~\cite{DesharnaisEA98}, and expressiveness for various
logics for trace-like equivalences on non-deterministic and
probabilistic systems. The emphasis in the examples has been on
well-researched equivalences and logics for the basic case of labelled
transition systems, aimed at demonstrating the versatility of graded
monads and graded logics along the axis of granularity of system
equivalence. The framework as a whole is however parametric also over
the branching type of systems and in fact over the base category
determining the structure of state spaces; an important direction for
future research is therefore to capture (possibly new) equivalences
and extract expressive logics on other system types such as
probabilistic systems (we have already seen probabilistic trace
equivalence as an instance; see~\cite{BonchiEA17} for a comparison of
some equivalences on probabilistic automata, which combine
probabilities and non-determinism) and nominal systems, e.g.\ nominal
automata~\cite{BojanczykEA14,SchroderEA17}. Moreover, we plan to
extend the framework of graded logics to cover also temporal logics,
using graded algebras of unbounded depth.

\bibliography{ltbtgradedmonads}

%
%
\iffull
\clearpage
\appendix

\section{Omitted Proofs and Details}


\subsection{Proof of \autoref{thm:depth-1-graded-monads-M0}}

The mentioned applicable instances of the graded monad laws
are the following: $(M_0,\unit,\mult^{00})$ is a monad,
$(M_1,\mult^{10})$ a (right) $M_0$-module, i.e.\ the following
diagrams commute:
  \[
    \begin{tikzcd}
      M_1 \ar{r}{M_1\unit} \arrow[dr,equal] & M_1M_0 \ar{d}{\mult^{10}}
          && M_1M_0M_0 \ar{r}{\mult^{10}M_0} \ar{d}{M_1\mult^{00}} &
          M_1M_0\ar{d}{\mult^{10}} %
        \\
        & M_1
          && M_1M_0\ar{r}{\mult^{10}} & M_1
    \end{tikzcd}
  \]
  $(M_1,\mult^{01})$ is a (left) $M_0$-module,
  \[
    \begin{tikzcd}
      M_1 \ar{r}{\unit M_1} \arrow[dr,equal] & M_0M_1 \ar{d}{\mult^{01}}
          && M_0M_0M_1 \ar{r}{M_0\mult^{01}} \ar{d}{\mult^{00}M_1} &
          M_0M_1\ar{d}{\mult^{01}} %
        \\
        & M_1
          && M_0M_1\ar{r}{\mult^{01}} & M_1
    \end{tikzcd}
  \]
  and $\mult^{10}$ is (componentwise) an $M_0$-algebra homomorphism,
  i.e.\ $\mu^{10}\cdot\mu^{01}M_0=\mu^{01}\cdot M_0\mu^{10}$:
  \[
    \begin{tikzcd}
      M_0M_1M_0 \ar{r}{M_0\mult^{10}}\ar{d}{\mult^{01}M_0} & M_0M_1 \ar{d}{\mult^{01}}\\
      M_1M_0 \ar{r}{\mult^{10}} & M_1
    \end{tikzcd}
  \]

\begin{proof}
  Note first that $\mult^{10}$ is a coequalizer in $\Set^{M_0}$, 
  in fact, a split coequalizer
  \[
    \begin{tikzcd}
      M_1M_0M_0
        \arrow[shift left=1]{r}[above]{M_1\mult^{00}}
        \arrow[shift right=1]{r}[below]{\mult^{10}M_0}
      &
      M_1M_0
        \arrow[bend left=50]{l}{M_1M_0\unit}
        \arrow{r}{\mult^{10}}
      &
      M_1
        \arrow[bend left=40]{l}{M_1\unit}
    \end{tikzcd}
  \]
  For every $k\geq 1$, we now define $M_k$, $\mult^{0,k+1}\colon M_0M_{k+1}\to M_{k+1}$ and
  $\mult^{1k}\colon M_1M_k\to M_{k+1}$ inductively by
  taking the coequalizer
  \begin{equation}\label{diag:depth-1-mult-1-k-coequalizer}
    \begin{tikzcd}
      M_1M_0M_k
        \arrow[shift left=1]{r}[above]{M_1\mult^{0k}}
        \arrow[shift right=1]{r}[below]{\mult^{10}M_k}
      &
      M_1M_k \arrow{r}{\mult^{1k}} & M_{k+1}
    \end{tikzcd}
  \end{equation}
  (objectwise) in $\Set^{M_0}$; 
  here we use that $M_1$ preserves $M_0$-algebras, in particular $\mult^{0k}$, and
  $\mult^{0,k+1}$ is given by the $M_0$-algebra structures of $M_{k+1}X$ for any $X$.
  Thus, we have
  \[
    \begin{tikzcd}
      M_{k+1}\ar{r}{\unit M_{k+1}}\ar[equal]{dr} & M_0M_{k+1} \ar{d}{\mult^{0,k+1}}
        && M_0M_0M_{k+1}\ar{r}{M_0\mult^{0,k+1}}\ar{d}{\mult^{00}M_{k+1}}
         & M_0M_{k+1}\ar{d}{\mult^{0,k+1}}
        \\
      & M_{k+1}
        && M_0M_{k+1}\ar{r}{\mult^{0,k+1}}
         & M_{k+1}
    \end{tikzcd}
  \]
  Note that the left-hand triangle above shows the unit laws for the
  $\mu^{0k}$. The remaining $\mult^{nk}$ for $n,k<\omega$ are again defined inductively, this time over
  $n$. In the induction step one uses the universal property of the coequalizer
  $\mult^{1n}M_k$ and that $\mult^{nk}$ is an $M_0$-algebra homomorphism:
  \begin{equation}\label{diag:depth-1-mult-n-k}
    \begin{tikzcd}[sep=large]
      M_1M_0M_nM_k
        \arrow[shift left=1]{r}[above]{M_1\mult^{0n}M_k}
        \arrow[shift right=1]{r}[below]{\mult^{10}M_nM_k}
        \arrow{d}{M_1M_0\mult^{nk}}
      & M_1M_nM_k
        \arrow[two heads]{r}{\mult^{1n}M_k}
        \arrow{d}{M_1\mult^{nk}}
      & M_{n+1}M_k
        \arrow[dashed]{d}{\mult^{n+1,k}}
      \\
      M_1M_0M_{n+k}
        \arrow[shift left=1]{r}[above]{M_1\mult^{0,n+k}}
        \arrow[shift right=1]{r}[below]{\mult^{10}M_{n+k}}
      & M_1M_{n+k}
        \arrow[two heads]{r}{\mult^{1,n+k}}
      & M_{n+k+1}
    \end{tikzcd}
  \end{equation}
  Indeed, since the two left-hand squares commute, we obtain a unique morphism on the
  right-hand edge making the right-hand square commutative.
  We verify the necessary properties of $\mult^{n+1,k}$
  \begin{longitemslist}
    \item For $k=0$ we have $\mult^{n+1,0}\cdot M_{n+1}\unit = \id$, using that
      $\mult^{1n}$ is an epimorphism:
      \[
        \begin{tikzcd}
          M_1M_n
            \arrow[bend right=60, equal]{dd}
            \arrow{d}{M_1M_n\unit}
            \arrow[two heads]{r}{\mult^{1n}}
          & M_{n+1}
            \arrow{d}{M_{n+1}\unit}
          \\
          M_1M_nM_0
            \arrow{r}{\mult^{1n}M_0}
            \arrow{d}{M_1\mult^{n0}}
          & M_{n+1}M_0
            \arrow{d}{\mult^{n+1,0}}
          \\
          M_1M_n
            \arrow[two heads]{r}{\mult^{1n}}
          & M_{n+1}
        \end{tikzcd}
      \]
      Indeed, the lower square commutes by the definition of $\mult^{n+1,0}$ and the upper one
      by naturality. Since the left-hand edge is the identity by induction hypothesis, and
      $\mult^{1n}$ is an epimorphism, the right-hand edge is the identity.
    \item The remaining associativity laws for $\mult^{nk}$ are proved
      as follows.
      \begin{longitemslist}
        \item For $\mult^{0k}$, $k\geq 2$, we need to show for every $m,l$ with $m+l=k$ that
          the following diagram commutes:
          \begin{equation}\label{depth-1-mult-0-l-m}
            \begin{tikzcd}
              M_0M_lM_m
                \arrow{r}{M_0\mult^{lm}}
                \arrow{d}{\mult^{0l}M_m}
              & M_0M_k
                \arrow{d}{\mult^{0k}}
              \\
              M_lM_m
                \arrow{r}{\mult^{lm}}
              & M_k
            \end{tikzcd}
          \end{equation}
          This holds since $\mult^{lm}$ is obtained by the universal property of a coequalizer
          in $M_0$-algebras, hence it is an $M_0$-algebra morphism.

        \item Now we show for every $k$ and $n\geq 1$ that for every $l,m$ with $l+m=n$
          the following diagram commutes:
          \[
            \begin{tikzcd}
              M_lM_mM_k
                \arrow{r}{\mult^{lm}M_k}
                \arrow{d}{M_l\mult^{mk}}
              & M_nM_k
                \arrow{d}{\mult^{nk}}
              \\
              M_lM_{m+k}
                \arrow{r}{\mult^{l,m+k}}
              & M_{n+k}
            \end{tikzcd}
          \]
          For $l=0$ this holds since $m=n$ and $\mult^{nk}$ is an $M_0$-algebra homomorphism.
          For $l\geq 1$, consider the following diagram:
          \[
            \begin{tikzcd}
              M_1M_{l-1}M_mM_k
                  \arrow{rrr}{M_1\mult^{l-1,m}M_k}
                  \arrow{ddd}{M_1M_{l-1}\mult^{mk}}
                  \arrow[two heads]{dr}{\mult^{1,l-1}M_mM_k}
              & & & M_1M_{n-1}M_k
                  \arrow{ddd}{M_1\mult^{n-1,k}}
                  \arrow{dl}[above left]{\mult^{1,n-1}M_k}
              \\
              & M_lM_mM_k
                \arrow{r}{\mult^{lm}M_k}
                \arrow{d}{M_l\mult^{mk}}
              & M_nM_k
                \arrow{d}{\mult^{nk}}
              & \\
              & M_lM_{m+k}
                \arrow{r}{\mult^{l,m+k}}
              & M_{n+k}
              & \\
              M_1M_{l-1}M_{m+k}
                  \arrow{rrr}[below]{M_1\mult^{l-1,m+k}}
                  \arrow{ur}[below right]{\mult^{1,l-1}M_{k+1}}
              & & & M_1M_{n-1+k}
                  \arrow{ul}{\mult^{1,n-1+k}}
            \end{tikzcd}
          \]
          Its left-hand part commutes by naturality of $\mult^{1,l-1}$, the upper, right-hand
          and lower parts commute by definition of $\mult^{n+1,k}$
          (\hyperref[diag:depth-1-mult-n-k]{Diagram~\eqref{diag:depth-1-mult-n-k}}) or by
          \hyperref[diag:depth-1-mult-1-k-coequalizer]{Diagram~\eqref{diag:depth-1-mult-1-k-coequalizer}}
          in case $l$, $n$ and $l$, respectively, are equal to $1$.
          Now proceed by induction on $n$: For $n=1$ we have $l=1$ and $m=0$; thus the outside
          commutes by the same argument as above,
          \hyperref[depth-1-mult-0-l-m]{Diagram~\eqref{depth-1-mult-0-l-m}}.
          In the induction step, the outside commutes by induction hypothesis.
          Thus, the desired inside commutes when precomposed by the epimorphism
          $\mult^{1,l-1}M_mM_k$, thus it commutes.
        \item It remains to prove for every $k$ and $n\geq 1$ that for $l$,$m$ with $l+m=k$ the
          following diagram commutes:
          \[
            \begin{tikzcd}
              M_nM_lM_m
                \arrow{r}{M_n\mult^{lm}}
                \arrow{d}{\mult^{nl}M_m}
              & M_nM_k
                \arrow{d}{\mult^{nk}}
              \\
              M_{n+l}M_m
                \arrow{r}{\mult^{n+l,m}}
              & M_{n+k}
            \end{tikzcd}
          \]
          This is done analogously, i.e.\ by induction on $n$:
          \[
            \begin{tikzcd}
              M_1M_{n-1}M_lM_m
                \arrow{rrr}{M_1M_{n-1}\mult^{lm}}
                \arrow{ddd}{M_1\mult^{n-1,l}M_m}
                \arrow[two heads]{dr}{\mult^{1,n-1}M_lM_m}
              & & & M_1M_{n-1}M_k
                \arrow{ddd}{M_1\mult^{n-1,k}}
                \arrow{dl}[above left]{\mult^{1,n-1}M_k}
              \\
              & M_nM_lM_m
                \arrow{r}{M_n\mult^{lm}}
                \arrow{d}{\mult^{nl}M_m}
              & M_nM_k
                \arrow{d}{\mult^{nk}}
              & \\
              & M_{n+l}M_m
                \arrow{r}{\mult^{n+l,m}}
              & M_{n+k}
              & \\
              M_1M_{n-1+l}M_m
                \arrow{rrr}[below]{M_1\mult^{n-1+l,m}}
                \arrow{ur}[below right]{\mult^{1,n-1+l}M_m}
              & & & M_1M_{n-1+k}
                \arrow{ul}{\mult^{1,n-1+k}}
            \end{tikzcd}
          \]
          Note that the upper part commutes for every $n$ by naturality of $\mult^{1,n-1}$.
          For $n=1$, the left-hand square and right-hand square commutes by
          \hyperref[diag:depth-1-mult-1-k-coequalizer]{Diagram~\eqref{diag:depth-1-mult-1-k-coequalizer}},
          and the lower square by \eqref{diag:depth-1-mult-n-k}. The outside commutes
          because $\mult^{lm}$ is an $M_0$-algebra homomorphism, thus the desired inner square
          commutes (when precomposed by the epimorphism $\mult^{1,n-1}M_lM_m$).
          For the induction step, the left-hand, right-hand and lower squares commute by the
          defining square \eqref{diag:depth-1-mult-n-k} and
          the outside by induction hypothesis. Again, the desired inner square commutes since
          $\mult^{1,n-1}M_lM_m$ is an epimorphism.\qedhere
      \end{longitemslist}
  \end{longitemslist}
\end{proof}

\subsection{Details for \autoref{E:d1}}
\label{S:d1}

\begin{proposition}\label{P:dist}
  Let $(T, \eta,\mu)$ be a monad and $F$ an endofunctor on $\cat
  C$. Then the graded monad $(M_n)_{n<\omega}$ of 
  \autoref{E:graded-monad}\ref{E:graded-monad:kleisli} is depth-1.
\end{proposition}
\begin{proof}
  Recall first that a \emph{split coequalizer} in some category is a
  diagram of the form
  \begin{equation}\label{eq:split}
    \xymatrix{
      A\ar@<3pt>[r]^-f \ar@<-3pt>[r]_-g
      &
      B \ar[r]^-{c}
      \ar@/^1.5pc/[l]^-t
      &
      C
      \ar@/^1pc/[l]^-s
    }
    \quad
    \text{satisfying the equations}\qquad
    \begin{array}{r@{\ }c@{\ }l}
      c\cdot f & = & c \cdot g \\
      c \cdot s & = & \id\\
      f\cdot t & = & \id \\
      g \cdot t & = & s \cdot c
    \end{array}
  \end{equation}
  It then follows that $c$ is indeed an (absolute) coequalizer of $f$
  and $g$.

  We will prove that~\eqref{eq:mu1n} is componentwise a split
  coequalizer in the category of $M_0$-algebras. Note first that here
  we have
  \begin{align*}
    \mu^{1n} &= (
    TFTF^n\xrightarrow{T\lambda F^n} TTF^{n+1} \xrightarrow{\mu
      F^{n+1}} TF^{n+1}) \\
    M_1\mu^{0n} &= (
    TFTTF^n \xrightarrow{TF\mu F^n} TFTF^n) \\
    \mu^{10}M_n &= (
    TFTTF^n \xrightarrow{T\lambda TF^n} TTFTF^n \xrightarrow{\mu
      FTF^n} TFTF^n) \\
  \end{align*}
  Note that all of these are componentwise $M_0$-algebra
  homomorphisms. Furthermore, so are the desired splittings $s$ and
  $t$, which are given by
  \begin{align*}
    TF\eta F^n\colon & TF^{n+1} = TFF^n \to TFTF^n
    \qquad\text{and}\\
    TFT\eta F^n\colon & TFTF^n \to TFTTF^n.
  \end{align*}
  We readily verify the four desired equations above. The first one
  holds by the laws of the graded monad $M_n$. For the second one use
  the unit laws of the distributive law and the monad $T$:
  \[
    \xymatrix{
      TFTF^n \ar[r]^-{T\lambda F^n} & TTF^{n+1} \ar[r]^-{\mu F^{n+1}}
      & TF^{n+1} \ar@{<-} `u[l] `[ll]_-{\mu^{1n}}[ll]\\
      &&
      TFF^n \ar@{=}[u]
      \ar[llu]^-{TF\eta F^n}
      \ar[lu]_-{T\eta FF^n}
    }
  \]
  The third equation again follows from one of the unit laws of the monad
  $T$: just apply $TF$ from the left and $F^n$ from the right on both
  sides of $\mu \cdot T\eta =  \id_T$.
  
  Finally, the last desired equation follows from naturality  
  \[
    \xymatrix@C+1pc{
      TFTF^n
      \ar[r]^-{T\lambda F^n}
      \ar[d]_{TFT\eta F^n}
      &
      TTFF^{n}
      \ar[r]^-{\mu F^{n+1}}
      \ar[d]_{TTF\eta F^n}
      &
      TFF^n
      \ar[d]^{TF\eta F^n}
      \ar@{<-} `u[l] `[ll]_-{\mu^{1n}}[ll]
      \\
      TFTTF^n
      \ar[r]_-{T\lambda TF^n}
      &
      TTFTF^n
      \ar[r]_-{\mu FTF^n}
      &
      TFTF^n
      \ar@{<-} `d[l] `[ll]^-{\mu^{10} M_n} [ll]
    }
  \]  
  This completes the proof.
\end{proof}

A related result concerns distributive laws of a monad over an
endofunctor (so called Eilenberg-Moore laws). Let $(T,\eta, \mu)$ be a monad and
$F$ an endofunctor on $\cat C$ equipped with a distributive law
$\lambda\colon TF \to FT$, i.e.~we have
\[
  \lambda \cdot \eta F = F\eta
  \qquad
  \text{and}
  \qquad
  \lambda \cdot \mu F = F \mu \cdot \lambda T \cdot T \lambda.
\]
Then, as shown in~\cite{MiliusEA15}, we obtain a graded monad
as follows: define $\lambda^n\colon TF^n \to F^nT$ inductively by
$\lambda^0 = \id_T$ and 
\[
  \lambda^{n+1} = (TF^{n+1} = TF^nF \xrightarrow{\lambda^nF} F^n
  TF\xrightarrow{F^n \lambda} F^nF T = F^{n+1} T);
\]
then we obtain a graded monad with $M_n = F^nT$, unit $\eta$ and
multiplication given by
\[
  \mu^{nk} = (F^nTF^kT \xrightarrow{F^n\lambda^k T} F^{n+k}TT
  \xrightarrow{F^{n+1}\mu} F^{n+k}T).
\]
\begin{proposition}
  The above graded monad is depth 1.
\end{proposition}
\begin{proof}
  The proof is similar to the one of \autoref{P:dist}. This
  time we show that~\eqref{eq:mu1n} is componentwise the coequalizer
  of a $U$-split pair, where $U\colon \Set^{M_0} \to \Set$ is the forgetful
  functor, i.e.\ the two splittings $s$ and $t$ are not required to be
  $M_0$-algebra homomorphisms. By Beck's theorem (see
  e.g.~\cite{MacLane98}), $U$ creates coequalizers of $U$-split
  pairs. Thus, it suffices to verify the equations in~\eqref{eq:split}
  in $\Set$. The calculations are completely analogous to what we have
  seen in the proof of \autoref{P:dist} using (the components of)
  \begin{align*}
    F\eta F^nT\colon& F^{n+1} T = FF^n T \to FTF^nT, \\
    F\eta TF^nT\colon& FTF^nT \to FTTF^nT
  \end{align*}
  as $s$ and $t$, respectively. We leave this as an easy exercise for
  the reader.\smnote{I'll add the details when I have time.}
\end{proof}

\subsection{Details for \autoref{sec:ltbt}}

\noindent To prove that a given description~$(M_n)_{n<\omega}$ of the
graded monad generated by a given graded theory~$\GTh$ is correct, we
generally proceed as follows.
\begin{itemize}
\item We identify a notion of \emph{normal form} of terms of a given
  depth~$n$, and show that every depth-$n$ term can be brought into
  this form, by a fixed \emph{normalization} procedure.
\item It will then typically be easy to see that depth-$n$ normal
  forms of terms over variables from a set~$X$ are in bijection with
  the claimed description~$M_nX$; the interpretation of the operations
  of~$\GTh$ over normal forms (by term formation and subsequent
  normalization) then transfers to $M_nX$ along this bijection. (This
  step also determines the multiplication, whose explicit description
  we otherwise mostly elide.)
\item We finally show that under this interpretation of the operations
  of~$\GTh$, the $M_nX$ form a graded algebra for~$\GTh$, i.e.\ satisfy its
  equations. This proves that no two distinct normal forms are
  provably equal, thus showing that $(M_n)_{n<\omega}$ is indeed the
  graded monad generated by~$\GTh$.
\end{itemize}
\noindent The description~$(M_n)_{n<\omega}$ of the graded monad then
usually makes it immediate that the associated graded semantics
captures the process equivalence at hand, in that the graded monad
contains precisely the data in the original semantics. In some cases
we need to introduce modifications of the original semantics that we
prove to induce the same notion of process equivalence, notably for
ready and failure traces and for readiness semantics.

\myparagraph{Details on \hyperref[ssec:completed-trace-semntics]{Completed Trace Semantics}}
Two states $x,y$ in labelled transition systems are \emph{completed
  trace equivalent} if $\CT(x) = \CT(y)$ where $\CT(x)$ is defined as
follows for a state $x$ in an LTS with carrier $X$
\cite{vanglabbeek2001linear}:
\begin{align*}
  \CT(x)& = \big\{
    w\in\Act^\ast \mid \exists z\in X\,.\,x\xrightarrow{w}z
    \big\}\\
  & \cup \big\{
    w\star\in\Act^\ast\times\{\star\}\mid\exists z\in X\,.\,x\overset w\to z\wedge
    I(z)=\varnothing
    \big\}.
\end{align*}
An element of $\CT(x)$ is thus either a trace or a \emph{complete
  trace}, the latter being recognizable by the marker~$\star$.

Normal forms in the graded theory for completed trace semantics as defined in
\autoref{sec:ltbt} are described as follows: A depth-$0$ term is
normalized by just collapsing nested joins into a set, identifying the
term with the set of elements that occur in it (in the sequel, we will
mostly keep normalization of joins into sets implicit). At depth
$n+1$, actions are distributed over the joins of depth-$n$ terms
(normalized by induction to contain joins only at the top level). We
end up with normal forms of depth-$n+1$ terms being joins of terms
consisting either of~$n+1$ unary operations of the form~$\sigma$ for
$\sigma\in\Act$, applied to a variable, or of at most~$n$ operations
of the form~$\sigma$, applied to the constant~$\star$; terms of the
latter kind represent complete traces, those of the former kind
represent standard traces.

Clearly, these normal forms are in bijection with our claimed
description
$M_nX=\Pfin(\Act^n\times
X+\Act^{<n}\times\{\star\})\cong\Pfin(\Act^n\times X+\Act^{<n})$.
The arising interpretation of the operations of the theory is as
follows: Joins are interpreted by set union; operations~$\sigma$ for
$\sigma\in\Act$ by prefixing all words in a set with~$\sigma$; and
the constant~$\star$ by the set~$\{\star\}$. In our standard programme
outlined above, it remains to prove that this graded algebraic
structure on the $M_nX$ satisfies the equations of the graded
theory. The join equations are inherited from powerset, and
distributivity of the unary operations~$\sigma$ over joins is clear by
the description of the interpretation of~$\sigma$ just given. This
proves that the graded monad generated by the graded theory of
complete traces is indeed
$M_nX=\Pfin(\Act^n\times X+\Act^{<n}\times\{\star\})$, and thus
represents complete traces. Formally, it remains to be shown that the
graded semantics induced by~$\alpha$ as described in
\autoref{sec:ltbt} does indeed compute the completed trace
semantics of a state (and not some other set of traces and complete
traces). We carry this proof out explicitly for the case at hand and
elide it in the sequel.

\begin{proposition}
  \label{prop:completed-trace-semantics-alpha}
  The $\nth$ stage $M_n!\cdot \gamma^{(n)}(x)$ of the $\alpha$-trace
  semantics of a state~$x$ in an LTS $(X,\gamma)$ contains exactly the
  length-$n$ traces of~$x$ and the complete traces of~$x$ of length
  less than~$n$.
\end{proposition}
(Consequently, $\alpha$-trace equivalence coincides with completed
trace equivalence.)
\begin{proof}
  We show the stronger statement that $\gamma^{(n)}(x)$ contains
  exactly the length-$n$ pretraces of~$x$ and the complete traces
  of~$x$ of length less than~$n$. We proceed by induction
  over~$n$. For $n=0$,~ $x$ has exactly one length-$n$ pretrace,
  namely $(\epsilon,x)$ (and, of course, no complete trace of length
  less than~$n$). On the other hand, we also have
  $\gamma^{(0)} (x) = \unit_X(x) = \{(\epsilon,x)\}$.
  In the inductive step from~$n$ to $n+1$, if $\gamma(x)=\varnothing$,
  then $x$ has no traces of depth $n+1$ and exactly one complete
  trace of depth $<n$, written $\star$. On the other hand,
  $\gamma^{n+1}(x)=\mu^{1n}\cdot
  M_1\gamma^{(n)}\cdot\alpha_X\cdot\gamma(x)$
  arises, in the algebraic view, by substituting into the term~$\star$
  that represents
  $\alpha\cdot\gamma(x) = \alpha(\varnothing) = \{\star\}$;
  since~$\star$ is a constant, it is left unchanged by substitutions,
  so that $\gamma^{(n+1)}(x)=\{\star\}$ as required. If
  $\gamma(x) \neq\varnothing$, then
  $\alpha(\gamma(x)) = \Pfin\inl(\gamma(x))$ is represented in
  algebraic notation as
  \begin{equation*}
    \alpha_X(\gamma(x)) = \sum_{(\sigma,x)\in\gamma(x)}\sigma(x).
  \end{equation*}
  By the description of the interpretation of the algebraic operations
  in~$\BBM$, we thus have
  \begin{align*}
    \gamma^{(n+1)}(x) & = \mu^{1n}\cdot M_1\gamma^{(n)}\cdot\alpha_X\cdot\gamma(x)\\
                      & = 
                        \big\{(\sigma w,z) \mid (\sigma,y)\in\gamma(x), (w,z)\in
                        \gamma^{(n)}(y)
                        \big\},
  \end{align*}
  whence the claim is immediate by the inductive hypothesis.
\end{proof}

\myparagraph{Details on \hyperref[ssec:readiness_failures]{Readiness and
Failures Semantics}}

The set of \emph{ready pairs} of a state $x$ of an LTS with carrier
$X$ is defined as
  \[
     \R(x)=\big\{ ( w,A)\in\Act^\ast\times\Pfin\Act\mid\exists z\in
     X\,.\,x\xrightarrow{w} z\wedge A=I(z) \big\},
  \]
and the set of \emph{failure pairs} of a state $x$ of an LTS with carrier $X$ is defined as
\[
  \F(x) = \big\{
    (w,A) \in \Act^\ast\times\Pfin\Act\mid\exists z\in X\,.\,x\xrightarrow{w}z
    \wedge A\cap I(z) = \emptyset
    \big\}.
\]
Since failures and readiness semantics are both finer than trace
semantics, i.e.\ two states that are readiness (failures) equivalent
are also trace equivalent~\cite{vanglabbeek2001linear}, we can add
traces $\T(x)$ to the sets $\R(x)$ and $\F(x)$ while still capturing
readiness (failures) equivalence; that is
$\R(x) = \R(y) \Leftrightarrow \R(x)\cup\T(x) = \R(y) \cup \T(y)$ in
case of readiness and analogously for failures.

We claim that the graded monad induced by the given graded theory of readiness has the form
\begin{align*}
  M_0X &= \Pfin X \\
  M_nX &= \Pfin(\Act^n\times X + \Act^{<n}\times\Pfin\Act)
\end{align*}
and the graded monad induced by the graded theory of failures
\begin{align*}
  M_0X &= \Pfin X \\
  M_nX &= \Down(\Act^n\times X + \Act^{<n}\times\Pfin\Act).
\end{align*}
The following description of normal forms will show that the graded equivalences induced by
these theories are exactly readiness and failures equivalence: $M_n1$ contains traces of length
$n$ as well as ready (failure) pairs of length less than $n$, i.e.\ the trace of the ready
(failure) pair is of length less than $n$. Normal forms over variables in $X$ are as follows.
Depth-$0$ terms are normalized by collapsing nested joins. Depth-$n+1$ terms are normalized by
distributing actions $\sigma$ over the top-level joins of (by induction, normalized) depth-$n$
terms. We end up with a join of depth-$n+1$ pretraces and up to depth-$n$ ready (failure)
pairs. In case of failures we also impose the order induced by the monotonicity on failure
sets, (co-)product ordering as implied by the signature of $M_n$, equality on $\Act^k$ and
inclusion order on $\Pfin\Act$, and normalize the terms by forming the downclosure accordingly.
Thus, the above descriptions of the graded monads are in bijection with these normal forms and
the $M_nX$ form graded algebras as the join semilattice equations, and in the case of failures
the monotonicity equations, are satisfied.

\myparagraph{Details on \hyperref[ssec:ready-failure-trace]{Ready Trace and
Failure Trace Semantics}}

The set $\RT(x)$ of ready traces in normal form of a state $x$ of an
LTS with carrier $X$ is defined as
  \begin{align*}
    \RT(x)=\bigl\{&
    A_0\sigma_1A_1\ldots \sigma_n A_n\in(\Pfin\Act\times\Act)^\ast\times\Pfin\Act\mid\\
    &\exists x_0,\dots,x_n\in X.\,x=x_0\xrightarrow{\sigma_1} x_1
    \xrightarrow{\sigma_2} \cdots \xrightarrow{\sigma_n} x_n\wedge \forall i\leq n\,.\,A_i=I(x_i)
    \bigr\}.
  \end{align*}
States $x,y$ are ready trace equivalent if $\RT(x) =\RT(y)$.

As indicated above, we work with an equivalent variant of this
semantics:
\begin{lemma}\label{lem:ready-traces}
  Ready trace equivalence coincides with the equivalence defined by
  assigning to a state $x$ of an LTS with carrier $X$ the set
  \begin{align*}
    \RT'(x) = \big\{&
      A_1\sigma_1\ldots A_n\sigma_n\in (\Pfin\Act\times\Act)^\ast \mid
      \exists x_0,\dots,x_n\in X.\,
      \\
      & x=x_0\xrightarrow{\sigma_1} x_1\xrightarrow{\sigma_2}\cdots\xrightarrow{\sigma_{n-1}} x_{n-1}\xrightarrow{\sigma_{n}} x_n
      \wedge  \forall i< n.\,A_i = I(x_i))
      \big\}\; \cup \\
      \big\{
        &A_1\sigma_1\ldots A_n\sigma_n\star\in (\Pfin\Act\times\Act)^\ast\times 1 \mid
        \exists x_0,\dots,x_n\in X.\,
        \\
        & x=x_0\xrightarrow{\sigma_1} x_1\xrightarrow{\sigma_2}\cdots
        \xrightarrow{\sigma_{n-1}} x_{n-1}\xrightarrow{\sigma_{n}} x_n
        \wedge\forall i< n.\, I(x_i)=A_i\,\wedge \\
        & I(x_n) = \emptyset
      \big\}.
  \end{align*}
\end{lemma}
\begin{proof}
  All the data in $\RT'(x)$ can clearly be calculated
  from~$\RT(x)$. Conversely, given $\RT'(x)$, a ready trace of~$x$ has
  the form $wA$ where $w\in\RT'(x)$. If $A=\emptyset$, then
  $w\star\in\RT'(x)$, so the information about~$wA$ is contained
  in~$\RT'(x)$. Otherwise pick $a\in A$; then $wAa\in\RT'(x)$, so
  again the information about $wA$ is in $\RT'(x)$. This proves the
  claim.
\end{proof}

\noindent The set $\FT(x)$ of failure traces in normal form of a state
$x$ of an LTS with carrier~$X$ is defined as
\begin{align*}
  \FT(x)=\big\{& A_0\sigma_1A_1\ldots \sigma_n A_n\in(\Pfin\Act\times\Act)^\ast
  \times\Pfin\Act\mid\\
  &\exists x_0,\dots,x_n\in
  X.\,x=x_0\xrightarrow{\sigma_1}x_1\xrightarrow{\sigma_2}\cdots\xrightarrow{\sigma_n}x_n\,
  \wedge \\
  &\forall i\leq n.\,A_i\cap I(x_i)=\varnothing\big\}
\end{align*}
Two states $x,y$ are failure trace equivalent if $\FT(x) =\FT(y)$.

\begin{lemma}\label{lem:failure-traces}
  Failure trace equivalence coincides with the equivalence defined by
  assigning to a state $x$ of an LTS with carrier $X$ the set
  \begin{align*}
    \FT'(x) =
      \big\{&
      A_1\sigma_1\ldots A_n\sigma_n\in (\Pfin\Act\times\Act)^\ast \mid
      \exists x_0,\dots,x_n\in X.
      \\
      & x=x_0\xrightarrow{\sigma_1} x_1\ldots x_{n-1}\xrightarrow{\sigma_{n}} x_n
      \wedge \forall i< n.\,A_i \cap I(x_i)) =\emptyset
      \big\} \;\cup \\
      \big\{
        &A_1\sigma_1\ldots A_n\sigma_n\star\in (\Pfin\Act\times\Act)^\ast\times 1 \mid
        \exists x_0,\dots,x_n\in X.
        \\
        & x=x_0\xrightarrow{\sigma_1} x_1\ldots x_{n-1} \xrightarrow{\sigma_{n-1}} x_n
        \;\wedge\\
            & \forall i<n.\, I(x_i)) \cap A_i =\emptyset \wedge I(z) = \emptyset
      \big\}
  \end{align*}
\end{lemma}
\begin{proof}
  All data in $\FT'(x)$ can clearly be calculated from $\FT(x)$
  (noting that a state $z$ is a deadlock iff
  $\Act\cap I(z)=\emptyset$). Conversely, given $\FT'(x)$ a failure
  trace in $\FT(x)$ has the form $wA$ where $w\in\FT'(x)$, and then
  there exists~$z$ such that $x\xrightarrow{w}z$ and
  $I(z)\cap A=\emptyset$ (where $x\xrightarrow{B}x$ whenever
  $I(x)\cap B=\emptyset$ for $B\subseteq\Act$). If~$I(z)=\emptyset$,
  then $w\star\in\FT'(x)$, and from this knowledge $wA\in\FT(x)$ is
  deducible. Otherwise, pick $a\in I(z)$. Then $wAa\in\FT'(x)$, from
  which knowledge $wA\in\FT(x)$ is deducible. This proves the claim.
\end{proof}

\begin{proposition}
  States in labelled transition systems are ready (failure) trace
  equivalent iff they have the same graded trace sequence for the given
  graded monads and graded trace semantics.
\end{proposition}
\begin{proof}
  For ready trace semantics, the claim is, by
  \autoref{lem:ready-traces}, shown in exactly the same way as for
  complete traces.

  For failure traces, one similarly uses the alternative failure trace
  semantics $\FT'$ of \autoref{lem:failure-traces}. We note that the
  graded monad induced by the graded theory defined in
  \autoref{sec:ltbt} is described as follows. We order $\Act$ by
  equality and $\Pfin\Act$ by set inclusion, and equip
  $(\Pfin\Act\times\Act)^n$ with the product ordering; on
  $(\Pfin\Act\times\Act)^{<n}=\sum_{i<n}(\Pfin\Act\times\Act)^i$, we use
  the coproduct ordering (in particular, words are comparable only if
  they have the same length). Then, $M_nX$ consists of the downclosed
  subsets of
  $(\Pfin\Act\times\Act)^n\times X+(\Pfin\Act\times\Act)^{<n}$.  One
  shows by induction over~$n$ that
  $A_1\sigma_1\dots A_n\sigma_n\in\FT'(x)$ iff
  $(A_1\sigma_1\dots A_n\sigma_n,\bullet)\in M_n!\cdot\gamma^{(n)}(x)$
  where now~$\bullet$ denotes the unique element of~$1$ (in this case
  indicating absence of information rather than deadlock), and, for
  $i<n$, $A_1\sigma_1\dots A_i\sigma_i\star\in\FT'(x)$ iff
  $A_1\sigma_1\dots A_i\sigma_i\in M_n!\cdot\gamma^{(n)}(x)$; again,
  details are like for complete traces. This proves the claim.
\end{proof}

\myparagraph{Details on \hyperref[ssec:simulation]{Simulation Equivalence}}

The description of the graded monad is seen as follows.  We define a
normal form for depth-$n$ terms inductively. The description of $M_0$
is standard, and arises by collapsing nested joins into sets. A
depth-$(n+1)$ term normalizes, in the same way, to a set of elements
of $\Act\times M_nX$ (with depth-$n$ terms normalized to elements of
$M_nX$ by induction). Monotonicity of the $\sigma\in\Act$ implies that
we can normalize such sets to be downclosed under the product ordering
on $\Act\times M_nX$; indeed, given a term
$\bigvee_{i \in I} \sigma_i(x_i)$ in $M_{n+1}X$ (corresponding to the
set $\{(\sigma_i, x_i) \mid i \in I\}$, we first expand the join by
adding for every $i \in I$ a copy of the summand $\sigma_i(x_i)$ for
every $y \leq x_i$, and then use monotonicity and the (infinitary)
congruence rule in the middle step below:
  \[
    \bigvee_{i \in I} \sigma_i(x_i) = \bigvee_{i \in I}\bigvee_{y \leq
      x_i} \sigma(x_i) = \bigvee_{i\in I}\bigvee_{y \leq x_i}
    \sigma_i(y) 
    = \bigvee_{i \in I, y \leq x_i} \sigma_i(y),
  \]
  which corresponds to a downward closed subset of $\Act \times
  M_nX$. 

  We arrive at normal forms that are in bijection with the claimed
  description of $M_{n+1}X$. The induced interpretations of the
  operations of the graded theory are as follows: Unary
  operations~$\sigma$ for $\sigma\in\Act$ map~$t\in M_nX$ to the
  downclosure of $\{(\sigma,t)\}$; and joins act as set unions. We are
  done once we show that under this interpretation, the $M_{n}X$ form
  a graded algebra, i.e.\ satisfy the equations of the graded
  theory. This is clear for the complete join semilattice equations;
  the monotonicity equation is ensured by the formation of
  downclosures in the interpretation of $\sigma\in\Act$.

For characterization of simulation by the graded monad, we prove the
stronger claim that given states $x$ and $y$ in $F$-coalgebras
$(X,\gamma)$ and $(Y,\delta)$, respectively,~$y$ simulates~$x$ up to
depth~$n$ iff $M_i!\cdot\gamma^{(i)}(x)\le M_i!\cdot\delta^{(i)}(y)$
for all $i\le n$ (a condition that is automatic for $i=0$). We proceed
by induction over~$n$, with trivial base case $n=0$. For the inductive
step, we note that~$y$ simulates~$x$ up to depth~$n+1$ iff for each
$(\sigma,x')\in\gamma(x)$ there exists $(\sigma,y')\in\delta(y)$ such
that~$y'$ simulates~$x'$ up to depth~$n$, and by induction, the latter
is equivalent to
$M_i!\cdot \gamma^{(i)}(x')\le M_i!\cdot \delta^{(i)}(y')$ for all
$i\le n$.  Thus, $y$ simulates~$x$ up to depth~$n$ iff for each
$i\le n$, the set
$\{(\sigma,M_{i}!\cdot \gamma^{(i)}(x'))\mid
(\sigma,x')\in\gamma(x)\}$
is contained in the downset of
$\{(\sigma,M_{i}!\cdot (\delta^{(i)}(y')))\mid
(\sigma,y')\in\delta(y)\}$
(w.r.t.\ the product ordering on $\Act\times M_i1$). Since for
$0<j\le n+1$, we can express $M_j!\cdot\gamma^{(j)}(x)$ in terms of
the graded theory as
\begin{equation*}
  \textstyle\bigvee_{(\sigma,x')\in\gamma(x)}\sigma(M_{j-1}!\cdot \gamma^{(j-1)}(x')),
\end{equation*}
correspondingly for~$y$, this condition is, by the definition of the
ordering on $M_i1$, equivalent to
$M_j!\cdot\gamma^{(j)}(x)\le M_j!\cdot\delta^{(j)}(y)$ for $j\le n+1$,
as desired.

\myparagraph{Details on \hyperref[ssec:GPS]{Probabilistic Trace Equivalence}}

\noindent Recall that the key equation in the graded theory for probabilistic
trace equivalence is 
\begin{equation}\label{eq:dist-dist}
  \sigma(x \boxplus_p y) = \sigma(x) \boxplus_p \sigma(y).
\end{equation}
In the following we use an equivalent description of elements of
$\dist(\Act\times X)$ as formal sums
\( \sum_{i\in I} p_i \sigma_i(x_i) \)
where $I$ is an index set, $p_i\in [0,1]$ with $\sum_{i\in I}p_i = 1$,
$\sigma_i\in\Act$, $x_i\in X$. Further, with $p(x, \sigma, z)$ we
denote the probability of a $\sigma$-transition from state $x$ to
$z$. The function $p$ can be extended to words over $\Act$ by the
following inductive definition:
\begin{gather*}
  p(x, \epsilon, y) = \left\{
    \begin{array}{ll}
      1 &\text{ if } x = y\\
      0 &\text{ otherwise}
    \end{array}
  \right.
  \\
  p(x, \sigma w, z) = \sum_{y\in X} p(x, \sigma, y) \cdot p(y, w, z)
  \text{ for } \sigma\in\Act,w\in\Act^\star
\end{gather*}

\begin{defn}
  Let $X$ be a GPS. The \emph{probabilistic trace semantics} of a
  state $x \in X$ is the following formal sum
  \[
    \T(x) = \sum_{\substack{w\in\Act^\star\\y\in X}} p(x,w,y) w,
  \]
  which is a probability distribution if the sum is restricted to
  words $w\in\Act^n$, i.e.\ words of a fixed length $n$.
\end{defn}

The graded monad induced by the depth-$1$ theory for probabilistic
trace semantics given in \autoref{ssec:GPS} has the unit
$\unit=\unit^{\dist}$ and the multiplication (by
\autoref{thm:depth-1-graded-monads-M0}):
\begin{gather*}
  \mult^{00} = \mult^\dist \\
  \mult^{01} = \dist\dist(\Act\times\Id)
  \xrightarrow{\mult^{\dist}(\Act\times\Id)}\dist(\Act\times\Id) \\
  \mult^{10} = \dist(\Act \times \dist) \xrightarrow{\dist\tau_{\Act,\Id}}
  \dist\dist(\Act\times\Id) \xrightarrow{\mult^{\dist}(\Act\times\Id)} \dist(\Act\times\Id)
\end{gather*}
where $\tau$ is the strength $\tau_{A,B}\colon A\times\dist B \to
\dist(A\times B)$ defined by
\[
  \tau\Big( a, \sum_{i\leq n}p_ib_i\Big) = \sum_{i\leq n} p_i(a, b_i).
\]
The description of normal forms of terms over $X$ is as follows: For
depth $0$ terms are simply elements of $\dist$. At depth $n+1$, a
distribution over pairs of actions and depth-$n$ normal forms is (by
induction) normalized by pulling the actions below the top-level
distribution of the depth-$n$ normal forms (by \eqref{eq:dist-dist})
and then applying the multiplication of the $\dist$-monad; the
resulting normal forms, elements of
$M_{n+1}X = \dist(\Act^{n+1}\times X)$, are distributions over pairs
of words of length $n+1$ and elements of $X$.
Formally, this is described by
\begin{align*}
  & \mult^{1n}_X \Big( 
    \sum_{\substack{\sigma\in\Act\\y\in X}} p(x,\sigma,y) \sigma\big( 
      \sum_{\substack{w\in\Act^n\\z\in X}} p(y,w,z) w
    \big)
  \Big)
  \\
  = & \sum_{\substack{\sigma w\in\Act^{n+1}\\z\in X}} \sum_{y\in X}
  p(x,\sigma,y)\cdot p(y,w,z) \sigma w
  \\
  = & \sum_{\substack{\sigma w\in\Act^{n+1}\\z\in X}}p(x,\sigma w,z) \sigma w.
\end{align*}

This shows that the $\nth$ component of the $\alpha$-trace sequence of a
state is the probability distribution over words of length $n$ at the
given state and thus $\alpha$-trace equivalence recovers probabilistic
trace equivalence.

\subsection{Details for \autoref{sec:logics}}

\subsubsection{Trace Invariance}

For standard labelled transition systems, a property of states (or a
formula expressing it) is \emph{trace invariant} if the property is
closed under trace equivalence. We denote by $\trdiamond{\sigma}$ the
usual Hennessy-Milner-style diamond operator for an action $\sigma$.
Then the formula $\trdiamond{\sigma}\top\land\trdiamond{\tau}\top$
expressing the property that both $\sigma$ and $\tau$ are one-letter
traces of the present state is clearly trace invariant. However, the
formula
$\trdiamond{\sigma}(\trdiamond{\sigma}\top\land\trdiamond{\tau}\top)$
is not trace invariant. Indeed, it is satisfied by every state that
proceeds with $\sigma$ to a state which has traces $\sigma$ and
$\tau$, see the left-hand picture below:
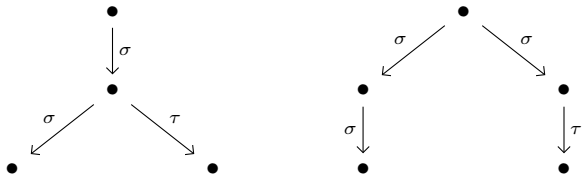
\[
  \begin{tikzcd}
    &\bullet \arrow{d}{\sigma}
    \\
    & \bullet
    \arrow{ld}[swap]{\sigma}\arrow{rd}{\tau}
    \\
    \bullet && \bullet
  \end{tikzcd}
  \qquad\qquad
  \begin{tikzcd}
    &
    \bullet
    \arrow{ld}[swap]{\sigma}\arrow{rd}{\sigma}\\
    \bullet \arrow{d}[swap]{\sigma} && \bullet \arrow{d}{\tau} \\
    \bullet && \bullet
  \end{tikzcd}
\]
However, the root state of the right-hand transition system is trace
equivalent to the root of the left-hand system, but does not satisfy
the formula. This demonstrates that conjunction cannot be included as
a logical operator in a trace invariant logic, which is supposed to be
compositional.

\subsubsection{Proof of \autoref{lem:canonical}}
\begin{proof}
  Regarding the reflexivity comment, note that
  $\mu^{10}\cdot M_1\eta_{A_0}=\id_{M_1A_0}=M_1a^{00}\cdot
  M_1\eta_{A_0}$ (independently of canonicity).

  \emph{`If':} Let $B$ be an
  $M_1$-algebra, and let $f\colon A_0\to B_0$ be a morphism of
  $M_0$-algebras. We have to show that $f$ extends uniquely to an
  $M_1$-algebra morphism $A\to B$. We have $M_0$-algebra morphisms
  \begin{equation*}
    \begin{tikzcd}
      M_1A_0 \arrow{r}{M_1f} & M_1B_0 \arrow{r}{b^{10}} & B_1
    \end{tikzcd}
  \end{equation*}
  whose composite $b^{10}\cdot M_1f$ coequalizes $\mu^{10}$ and $M_1a^{00}$:
  \begin{align*}
    b^{10}\cdot M_1f\mu^{10} & = b^{10}\cdot \mu^{10}\cdot M_1M_0f && \by{naturality}\\
                       & = b^{10}\cdot M_0b^{00}\cdot M_1M_0f && \by{$M_1$-algebra}\\
                       & = b^{10}\cdot M_1f\cdot M_1a^{00} && \by{$M_0$-morphism.}
  \end{align*}
  By the coequalizer property, we thus obtain an $M_0$-morphism
  $f^\sharp\colon A_1\to B_1$ such that 
  \begin{equation*}
    \begin{tikzcd}
      M_1A_0 \arrow{r}{M_1 f} \arrow{d}[left]{a^{10}} & M_1B_0\arrow{d}{b^{10}}\\
      A_1 \arrow{r}[below]{f^\sharp} & B_1
    \end{tikzcd}
  \end{equation*}
  commutes; since moreover both~$f$ and $f^\sharp$ are $M_0$-algebra
  morphisms, this implies that $(f,f^\sharp)$ is an $M_1$-algebra
  morphism $A\to B$, and clearly the unique such morphism
  extending~$f$.

  \emph{`Only if':} Let $B$ be an $M_0$-algebra, and let
  $f\colon M_1A_0\to B$ be an $M_0$-algebra morphism such that
  $f\cdot \mu^{01}=f\cdot M_1a^{00}$. It is then immediate from the assumptions
  that $\bar B=(a^{00},b^{00},f)$ is an $M_1$-algebra (with carriers
  $\bar B_0=A_0$, $\bar B_1=B_0$). By canonicity of~$A$, there is a
  unique $M_0$-algebra morphism $g\colon A_1\to B_0$ such that the pair
  $(id_{A_0},g)$ forms an $M_1$-algebra morphism $A\to\bar B$:
  \begin{equation*}
    \begin{tikzcd}
      M_1 A_0 \arrow{r}{M_1\id_{A_0}} \arrow{d}[left]{a^{10}} &
      M_1 A_0 \arrow{d}{f} \\
      A_1 \arrow{r}{g} & B_0.
    \end{tikzcd}
  \end{equation*}
  This shows that $f$ factorizes uniquely through~$a^{10}$, proving
  the desired coequalizer property of~$a^{10}$.
\end{proof}

\subsubsection{Truth functions for $M_0 = \Pfin$ and $\Omega = 2$}

The $M_0$-morphisms $\Omega^k \to \Omega$ are the join-continuous
functions. We show that such functions are disjunctions of some of
their arguments. For $k = 0$, the only join-continuous function is
$\bot\colon 1 \to \Omega$ (preserving the empty join). For $k = 1$ we
have the constant function on $\bot$ and $\id\colon \Omega \to \Omega$
corresponding to the disjunction of none and its only argument,
respectively. For $k = 2$, a join-continuous function
$t\colon \Omega^2 \to \Omega$ is either constant on $\bot$, or we have
$t(\top,\top) = \top$, and then $t(\top,\bot)$, or $t(\bot,\top)$, or
both must be $\top$, corresponding to the disjunction of only the
first, only the second, or both arguments of $t$, respectively. In
general, a join-continuous function $t\colon \Omega^k \to \Omega$ is
completely determined by the values on the join-irreducibles in
$\Omega^k$ (i.e.\ those tuples $\bar b \in \Omega^k$ with precisely one
$\top$-component). So these join-irreducibles correspond to the arguments
of $t$, and $t$ is in fact the disjunction of those of its
arguments whose corresponding join-irreducibles $\bar b$ satisfy
$t(\bar b) = \top$.  

\subsection{Details for \autoref{sec:expr}}

\subsubsection{Full Proof of \autoref{thm:expr}}

  For readability, we restrict to unary modalities.  We have to show
  that for each~$n$, the set of evaluation functions
  $\Sem{\phi}\colon M_n1\to\Omega$ of depth-$n$ trace formulae~$\phi$ is
  jointly injective. We proceed by induction on~$n$. The base case
  $n=0$ is immediate by depth-$0$ separation. For the induction step
  from~$n$ to $n+1$, let~$\FA$ denote the set of evaluation maps
  $M_n1\to\Omega$ of depth-$n$ trace formulae. By the inductive
  hypothesis, $\FA$ is jointly injective. Moreover, by the
  construction of the logic,~$\FA$ is closed under all propositional
  operators in~$\PropOps$. By depth-$1$ separation, it follows that
  the set
  \begin{equation*}
    \{L(\Sem{\phi})\mid L\in\Lambda,
    \phi\text{ a depth-$n$ formula}\}
  \end{equation*}
  is jointly injective. These maps are the interpretations of
  depth-$(n+1)$ trace formulae of the form $L\phi$, which proves the
  inductive claim.

\subsubsection{Details for \autoref{expl:bisim-logics}}

Let~$F$ be a finitary set
    functor. Then the graded monad $M_nX=F^nX$ captures behavioural
    equivalence in $F$-coalgebras. In this setting, $M_0$-algebras are
    just sets, and modalities in the arising graded logic (again, unary
    to save on notation) are maps $L\colon F2\to 2$, which are equivalent to
    \emph{predicate liftings}, i.e.\ natural transformations
    $\contrapow\to\contrapow\circ F^\Op$~\cite{Schroder08}. Combining
    a set~$\Lambda$ of such modalities with full Boolean propositional
    logic leads to \emph{coalgebraic modal
      logic}~\cite{Pattinson04,Schroder08}. It is well-known that the
    coalgebraic modal logic determined by~$\Lambda$ is
    \emph{expressive}, i.e.\ distinguishes behaviourally inequivalent
    states in $F$-coalgebras, if the set~$\Lambda$ of modalities is
    \emph{separating}, i.e.\ each element $t\in FX$ is uniquely
    determined by the set
    $\{(L,f)\in \Lambda\times 2^X\mid L(Ff(t))=\top\}$.

    This fact becomes an instance of \autoref{thm:expr} as
    follows. An $M_1$-algebra is just a map of the form $FA_0\to A_1$,
    and canonical algebras are coequalizers of two identical maps into
    $FA_0$, hence isomorphisms, and then w.l.o.g.\ identities
    $\id\colon FA_0\to FA_0 = A_1$; for a map $f\colon A_0\to 2$ (since $M_0$-algebras
    are just sets, all maps are $M_0$-homomorphisms), we then have
    $\Sem{L}(f)=L\cdot Ff\colon FA_0\to 2$. To see that separation in the sense recalled
    above implies depth-$1$ separation, let $\FA$ be a jointly
    injective set of maps $A_0\to 2$, closed under all (finitary)
    Boolean operations.  We have to show that the set of maps
    $\{\Sem{L}(f)\mid L\in\Lambda,f\in\FA\}$ is again jointly injective. Let
    $t,s$ be distinct elements of $FA_0$. Since~$F$ is finitary, there
    exists a finite $X\subseteq A_0$ and (distinct) $s',t'\in FX$ such
    that $s=Fi(s')$, $t=Fi(t')$ where~$i$ is the injection
    $X\into A_0$. Since~$\Lambda$ is separating, we have $L\in\Lambda$
    and $f\colon X\to 2$ such that $\Sem{L}(f)(s')\neq \Sem{L}(f)(t')$.  Since $\FA$ is
    jointly injective and closed under Boolean operations, and~$X$ is
    finite, there exists $g\in\FA$ such that $f=g\cdot i$. Then $L(g)$
    separates~$s$ and~$t$.

\subsubsection{Details for \autoref{expl:ltbt-logics}}

\subsubsection*{Trace equivalence} Let~$A$ be a canonical
$M_1$-algebra; then the structure map $\Pfin(\Act\times A_0)\to A_1$ is
surjective, i.e.\ every element of~$A_1$ has the form
$\bigvee_{i\in I}\sigma_i(x_i)$, for $\sigma_i\in\Act$, $x_i\in
A_0$.
Since the operations~$\sigma$ are complete join semilattice morphisms,
we can in fact write every element of~$A_1$ in the form
$\bigvee_{\sigma\in\Act}\sigma (x_\sigma)$.

Now, to show depth-$1$ separation, suppose we have two distinct
elements of~$A_1$; by the above, these have the form
$x=\bigvee_{\sigma\in\Act}\sigma(x_\sigma)$ and
$y=\bigvee_{\sigma\in\Act}\sigma(y_\sigma)$, respectively, and thus
there must exist $\sigma\in\Act$ such that $x_\sigma\neq y_\sigma$;
w.l.o.g.\ $x_\sigma\not\le y_\sigma$.  Since the $f\in\FA$ preserve
joins, joint injectivity of~$\FA$ thus implies that there exists
$f\in\FA$ such that $f(x_\sigma)=\top$ and $f(y_\sigma)=\bot$. (To see
this, note that $x_\sigma\not\le y_\sigma$ implies that
$x_\sigma\vee y_\sigma\neq y_\sigma$, so there exists $f\in\FA$ such
that $f(x_\sigma\vee y_\sigma)\neq f(y_\sigma)$, and by monotonicity
of~$f$, we must have $f(x_\sigma\vee y_\sigma)=\top$ and
$f(y_\sigma)=\bot$. But then
$f(x_\sigma)=f(x_\sigma)\vee \bot = f(x_\sigma)\vee f(y_\sigma)=
f(x_\sigma\vee y_\sigma)=\top$.)
Now recall that the modal operator
$\Sem{\trdiamond{\sigma}}\colon\Pfin(\Act\times 2)\to 2$ is defined by
$\Sem{\trdiamond{\sigma}}(S)=\top$ if $(\sigma,\top)\in S$, and
$\Sem{\trdiamond{\sigma}}(S)=\bot$ otherwise. The commutativity of
\[
\begin{tikzcd}[column sep = 20mm]
  P(\Act \times A_0)
  \arrow{r}{\Pfin(\Act \times f)}
  \arrow{d}[left]{a^{10}}
  &
 \Pfin(\Act \times \Omega) \arrow{d}{\Sem{\trdiamond{\sigma}}}
  \\
  A_1 \arrow{r}{\Sem{\trdiamond{\sigma}}(f)}
  & A_0
\end{tikzcd}
\]
(an instance of~\eqref{diag:L(f)}) yields that, for
$\bigvee_{\sigma\in\Act}\sigma(z_\sigma)$ in~$A_1$, we have
\begin{equation*}
  \Sem{\trdiamond{\sigma}}(f)(\textstyle\bigvee_{\sigma\in\Act}\sigma(z_\sigma))=
  \Sem{\trdiamond{\sigma}}\big(
  \{(\sigma, f(z_\sigma)) \mid \sigma \in \Act\}
  \big) =
  \begin{cases}
    \top & \text{ if $f(z_\sigma)=\top$}\\
    \bot & \text{ otherwise.}
  \end{cases}
\end{equation*}
Thus, $\Sem{\trdiamond{\sigma}}(f)$ separates~$x$ and~$y$.



\subsubsection{Details for \autoref{expl:prob-trace}}

We present the proof of depth-$1$ separation in detail. So we assume
given a canonical $M_1$-algebra with carriers $A_0,A_1$, which we view
mainly in terms of its algebraic operations, and a jointly injective
set $\FA$ of convex algebra morphisms $A_0\to[0,1]$ that is closed
under convex algebra morphisms, i.e.\ affine maps; in fact, we need
only that~$\FA$ contains the constant map~$1$. Again, canonicity is
used only to obtain that every element of~$A_1$ can be written as a
convex combination of elements of the form $\sigma(x)$ for $x\in A_0$.
We normalize convex combinations to mention every element of the base
set exactly once (maybe with coefficient~$0$), and then write then as
distributions in the usual way. So take two distributions $\mu,\nu$ on
$\Act\times A_0$, and suppose that
$\Sem{\langle\sigma\rangle}(f)(\mu)=\Sem{\langle\sigma\rangle}(f)(\nu)$
for each $f\in\FA$ and each $\sigma\in\Lambda$; we have to show
$\mu=\nu$. Applying the assumption to $f=1$, we obtain that for
each~$\sigma$,
\begin{equation*}\label{eq:per-sigma}\textstyle
  \sum_{x\in A_0}\mu(\sigma,x)=\sum_{x\in A_0}\nu(\sigma,x)\eqqcolon c_\sigma.
\end{equation*}
Of course, $\sum_{\sigma\in\Act}c_\sigma=1$.  If $c_\sigma=0$, then
$\mu(\sigma,x)=\nu(\sigma,x)=0$ for all~$x$, so it suffices to
consider $\sigma$ with $c_\sigma>0$. Let $f\in\FA$. By assumption, we
have
\begin{equation*}\textstyle
  \sum_{x\in A_0}\mu(\sigma,x)f(x)=\sum_{x\in A_0}\nu(\sigma,x)f(x)
\end{equation*}
and hence
\begin{align*}\textstyle
  f(\sum_{x\in A_0}c_\sigma^{-1}\mu(\sigma,x)x) & = \textstyle\sum_{x\in A_0}c_\sigma^{-1}\mu(\sigma,x)f(x)\\
  & = \textstyle c_\sigma^{-1}\sum_{x\in A_0}\mu(\sigma,x)f(x) \\
  & = \textstyle c_\sigma^{-1}\sum_{x\in A_0}\nu(\sigma,x)f(x) \\
  & = \textstyle\sum_{x\in A_0}c_\sigma^{-1}\nu(\sigma,x)f(x)\\
  & = \textstyle f(\sum_{x\in A_0}c_\sigma^{-1}\nu(\sigma,x)x),
\end{align*}
using in the first and last step that~$f$ is a morphism of convex
algebras and the coefficients of the sum add up to~$1$ after the
normalization with~$c_\sigma^{-1}$. Since $\FA$ is jointly injective, we obtain
\begin{equation*}\textstyle
  \sum_{x\in A_0}c_\sigma^{-1}\mu(\sigma,x)x = \sum_{x\in A_0}c_\sigma^{-1}\nu(\sigma,x)x.
\end{equation*}
Applying the operation $\sigma$ to both sides and using that~$\sigma$
is a morphism of convex algebras, we obtain
\begin{equation*}\textstyle
  \sum_{x\in A_0}c_\sigma^{-1}\mu(\sigma,x)(\sigma(x)) = \sum_{x\in A_0}c_\sigma^{-1}\nu(\sigma,x)(\sigma(x)),
\end{equation*}
an equation in~$A_1$. Taking the convex combination of these equations
with coefficients $c_\sigma$, we obtain $\mu=\nu$, as required. \qed
\fi

\end{document}